\documentclass{scrartcl}
\usepackage[left=2cm,right=2cm,top=3cm]{geometry}
\usepackage{amsmath}
\usepackage{amsthm}
\usepackage{amssymb}
\usepackage{amsfonts}
\usepackage{mathrsfs}
\usepackage{color}
\usepackage[utf8]{inputenc}
\usepackage{enumitem}
\usepackage{multicol}
\usepackage[english]{babel}
\usepackage[babel]{csquotes}
\usepackage{mathtools}
\usepackage{rotating}%for landsacape page
\usepackage{latexsym}
\usepackage{epsfig}
\usepackage{verbatim}
\usepackage{thmtools,thm-restate}
\usepackage{listings} %Per inserire codice
\usepackage{makecell,booktabs,amsxtra}%for hard trees

\usepackage{adjustbox}%to scale table size 
\usepackage{ifthen}
\usepackage[position=top]{subfig}
\usepackage{pgfplots}
\pgfplotsset{xticklabels={}, yticklabels={} width=7cm,compat=1.8}
\usepackage{pgfplotstable}
\pgfmathsetseed{1138} % set the random seed
\pgfplotstableset{ % Define the equations for x and y
	create on use/x/.style={create col/expr={42+2*\pgfplotstablerow}},
	create on use/y/.style={create col/expr={(0.6*\thisrow{x}+130)+5*rand}}
}
\pgfplotstablenew[columns={x,y}]{30}\loadedtable

\usepackage[ruled]{algorithm2e}

\SetAlFnt{\small}
\SetAlCapFnt{\small}
\SetAlCapNameFnt{\small}
\SetAlCapHSkip{0pt}
\IncMargin{-\parindent}

\usepackage{url}
\usepackage[square,numbers,sort]{natbib}

\usepackage{pgfplots,pgfplotstable}%stacked bar charts
\pgfplotsset{compat=1.8}

\usepackage[bookmarks=false,colorlinks=true,linkcolor=blue,urlcolor=blue,citecolor=red,pdfborder={0 0 0.5}]{hyperref}

\theoremstyle{plain}
\newtheorem{theorem}{Theorem}[section]
\newtheorem{lemma}[theorem]{Lemma}

\newtheorem{corollary}[theorem]{Corollary}
\newtheorem{conjecture}[theorem]{Conjecture}

\declaretheoremstyle[
notefont=\bfseries, notebraces={}{},
bodyfont=\normalfont\itshape,
headformat=\NAME \NOTE
]{nopar}
\declaretheorem[style=nopar,name=Corollary]{corollary*}
\declaretheorem[style=nopar,name=Theorem]{theorem*}

\theoremstyle{definition}
\newtheorem{definition}[theorem]{Definition}

\newtheorem*{remark*}{Remark}
\newtheorem{observation}[theorem]{Observation}

\newtheorem*{acknowledgements*}{Acknowledgements}

\newcount \mycount
\newcount \mycountb

\newcommand{\connect}[2]{\ensuremath{
{\leftrightsquigarrow}_{#2}}}

\DeclareMathOperator{\N}{\mathbb N}

\DeclareMathOperator{\bA}{{\mathbb A}}
\DeclareMathOperator{\A}{{\bA}}
\DeclareMathOperator{\B}{{\mathbb B}}
\DeclareMathOperator{\bB}{{\B}}
\DeclareMathOperator{\bP}{{\mathbb P}}
\DeclareMathOperator{\bC}{{\mathbb C}}
\DeclareMathOperator{\bD}{{\mathbb D}}
\DeclareMathOperator{\bF}{{\mathbb F}}

\DeclareMathOperator{\bT}{{\mathbb T}}
\DeclareMathOperator{\T}{{\bT}}

\DeclareMathOperator{\bI}{{\mathbb I}}

\renewcommand{\P}{\mathbb P}

\DeclareMathOperator{\Pol}{Pol}

\DeclareMathOperator{\Csp}{CSP}

\newcommand{\HM}[1]{\operatorname{HM}(#1)}
\newcommand{\Ele}[1]{\operatorname{Ele}(#1)}
\newcommand{\stCon}{{\operatorname{st-Con}}}

\providecommand{\dotdiv}{% Don't redefine it if available
  \mathbin{% We want a binary operation
    \vphantom{+}% The same height as a plus or minus
    \text{% Change size in sub/superscripts
      \mathsurround=0pt % To be on the safe side
      \protect\ooalign{% Superimpose the two symbols
        \noalign{\kern-.45ex}% but the dot is raised a bit
        \hidewidth$\smash{\cdot}$\hidewidth\cr % Dot
        \noalign{\kern.45ex}% Backup for vertical alignment
        $-$\cr % Minus
      }%
    }%
  }%
}
\providecommand{\cupdot}{% Don't redefine it if available
  \mathbin{% We want a binary operation
    \vphantom{+}% The same height as a plus or minus
    \text{% Change size in sub/superscripts
      \mathsurround=0pt % To be on the safe side
      \protect\ooalign{% Superimpose the two symbols
        \noalign{\kern-.4ex}% but the dot is raised a bit
        \hidewidth$\smash{\cdot}$\hidewidth\cr % Dot
        \noalign{\kern.4ex}% Backup for vertical alignment
        $\cup$\cr % Minus
      }%
    }%
  }%
}
\theoremstyle{definition}
\newtheorem{examplex}[theorem]{Example}
\newenvironment{example}
{\pushQED{\qed}\examplex}
{\popQED\endexamplex}

\theoremstyle{definition}

\makeatletter
\newcommand\footnoteref[1]{\protected@xdef\@thefnmark{\ref{#1}}\@footnotemark}
\makeatother

\makeatletter
\newcommand*{\rom}[1]{\expandafter\@slowromancap\romannumeral #1@}
\makeatother

\makeatletter
\def\blfootnote{\xdef\@thefnmark{}\@footnotetext}
\makeatother

\usepackage{setspace}

\usepackage{epigraph}
\usepackage{siunitx}
\newcounter{anzahl}
\def\nodeDist{0.4}

\newcommand{\tikzpathFreeVariables}[1]{
\begin{tikzpicture}[scale=0.5]
\setcounter{anzahl}{0}

\foreach \y/\c [count=\xi from 1] in #1 {
    \ifthenelse{\c=0}{
    \node[var-f] (\xi) at (\nodeDist*\xi,\nodeDist*\y) {};}{
    \node[var-b] (\xi) at (\nodeDist*\xi,\nodeDist*\y) {};}
    
    \setcounter{anzahl}{\xi}
    }
    
\foreach \y [count=\yi from 1,count=\yii from 2] in {2,...,\arabic{anzahl}} 
    \path (\yi) edge (\yii);
\end{tikzpicture}
}

\newcommand{\toEdge}{\mathbin
{\begin{tikzpicture}[baseline =-1mm]
    \path[->] (0,0) edge (0.37,0);
\end{tikzpicture}}}
\newcommand{\fromEdge}{\mathbin
{\begin{tikzpicture}[baseline =-1mm]
    \path[<-] (0,0) edge (0.37,0);
\end{tikzpicture}}}
\tikzstyle{var-b}=[circle,fill,draw=white,inner sep=0pt,minimum size=3.5pt]
\tikzstyle{bullet}=[circle,fill,draw=white,inner sep=0pt,minimum size=3.2pt]
\tikzstyle{var-f}=[circle,draw,inner sep=0pt,minimum size=3pt]
\usetikzlibrary{arrows,arrows.meta}
\usetikzlibrary{bending,calc}

\newcommand{\HMThree}{
    \begin{tikzpicture}
        \node at (0,1) {$p_0$};
        \node (0) at (0,0) {$x$};
        
        \node at (2,1) {$p_1$};
        \node (1) at (2,0) {$\begin{array}{cc}
             x&x  \\
             y&x\\
             y&y
        \end{array}$};
        
        \node at (4,1) {$p_2$};
        \node (2) at (4,0) {$\begin{array}{cc}
             x&x  \\
             y&x\\
             y&y
        \end{array}$};
        
        \node at (6,1) {$p_3$};
        \node (3) at (6,0) {$\begin{array}{cc}
             x&x  \\
             y&x\\
             y&y
        \end{array}$};
        
        \node at (8,1) {$p_4$};
        \node (4) at (8,0) {$y$};

        \path 
            (0) edge (1)
            (1) edge (2)
            (2) edge (3)
            (3) edge (4)
            ;

        \node[var-b] (braid0a) at (1,-1.5) {};
        \node[var-b] (braid0b) at (1,-2.5) {};
        
        \node[var-b] (braid1a) at (3,-1.5) {};
        \node[var-b] (braid1b) at (3,-2.5) {};
        
        \node[var-b] (braid2a) at (5,-1.5) {};
        \node[var-b] (braid2b) at (5,-2.5) {};
        
        \node[var-b] (braid3a) at (7,-1.5) {};
        \node[var-b] (braid3b) at (7,-2.5) {};
        \path[->,>=stealth']
            (braid0a) edge (braid1a)
            (braid0b) edge (braid1a)
            (braid0b) edge (braid1b)
            
            (braid1a) edge (braid2a)
            (braid1b) edge (braid2a)
            (braid1b) edge (braid2b)
            
            (braid2a) edge (braid3a)
            (braid2b) edge (braid3a)
            (braid2b) edge (braid3b)
        ;
    \end{tikzpicture}
}

\newcommand{\EleThree}{
    \begin{tikzpicture}
        \node at (0,1.5) {$e_0$};
        \node (0) at (0,0) {$\begin{array}{c}
             
             x_1\\
             x_2\\
             x_3
        \end{array}$};
        
        \node at (2,1.5) {$e_1$};
        \node (1) at (2,0) {$\begin{array}{cc}
             x_1&y  \\
             z&y\\
             z&x_1\\
             x_2&x_2\\
             x_3&x_3
        \end{array}$};
        
        \node at (4,1.5) {$e_2$};
        \node (2) at (4,0) {$\begin{array}{cc}
             x_1&x_1  \\
             x_2&y\\
             z&y\\
             z&x_2\\
             x_3&x_3
        \end{array}$};
        
        \node at (6,1.5) {$e_3$};
        \node (3) at (6,0) {$\begin{array}{cc}
             x_1&x_1  \\
             x_2&x_2\\
             x_3&y\\
             z&y\\
             z&x_3
        \end{array}$};
        
        \node at (8,1.5) {$e_4$};
        \node (4) at (8,0) {$\begin{array}{cc}
             x_1\\
             x_2\\
             x_3\\
        \end{array}$};
        \path 
            (0) edge (1)
            (1) edge (2)
            (2) edge (3)
            (3) edge (4)
            ;

        \node[var-b] (braid0a) at (1,-1.5) {};
        \node[var-b] (braid0b) at (1,-2.5) {};
        
        \node[var-b] (braid1a) at (3,-1.5) {};
        \node[var-b] (braid1b) at (3,-2.5) {};
        
        \node[var-b] (braid2b) at (5,-2.5) {};
        
        \node[var-b] (braid3b) at (7,-2.5) {};
        \path[->,>=stealth']
            (braid0a) edge (braid1a)
            (braid0b) edge (braid1a)
            (braid0b) edge (braid1b)
            
            (braid1b) edge (braid2b)
            
            (braid2b) edge (braid3b)
        ;

        \node[var-b] (braid0a) at (1,-3.0) {};
        
        \node[var-b] (braid1a) at (3,-3.0) {};
        \node[var-b] (braid1b) at (3,-4.0) {};
        
        \node[var-b] (braid2a) at (5,-3.0) {};
        \node[var-b] (braid2b) at (5,-4.0) {};
        
        \node[var-b] (braid3b) at (7,-4.0) {};
        \path[->,>=stealth']
            (braid0a) edge (braid1a)
            
            (braid1a) edge (braid2a)
            (braid1b) edge (braid2a)
            (braid1b) edge (braid2b)
            
            (braid2b) edge (braid3b)
        ;

        \node[var-b] (braid0a) at (1,-4.5) {};
        
        \node[var-b] (braid1a) at (3,-4.5) {};
        
        \node[var-b] (braid2a) at (5,-4.5) {};
        \node[var-b] (braid2b) at (5,-5.5) {};
        
        \node[var-b] (braid3a) at (7,-4.5) {};
        \node[var-b] (braid3b) at (7,-5.5) {};
        \path[->,>=stealth']
            (braid0a) edge (braid1a)
            
            (braid1a) edge (braid2a)
            
            (braid2a) edge (braid3a)
            (braid2b) edge (braid3a)
            (braid2b) edge (braid3b)
        ;
    \end{tikzpicture}
}

\tikzset{
	zigzagedge/.style={
		decoration={zigzag, amplitude=0.3mm, segment length=1.2mm}
	}
}

\newcommand{\figureBlocks}{
\begin{tikzpicture}[scale=0.55,rotate=-90]
		\def\colA{black}
		\def\colB{green!85!black}
		\def\colC{orange!90!black}
		\def\colD{blue}
		\def\colE{cyan!80!black}
		\def\colF{red}
		\node[\colA] (13a) at (-1,1) {$y$};
		\node[\colA] (14a) at (-3,1) {$x$};
		
		\node[\colA] (11b) at (0,3) {$y$};
		\node[\colA] (12b) at (-1,3) {$y$};
		\node[\colA] (13b) at (-2,3) {$y$};
		\node[\colA] (14b) at (-3,3) {$x$};

		\node[\colA] (21a) at (0,5) {$y$};
		\node[\colA] (22a) at (-1,5) {$y$};
		\node[\colA] (23a) at (-2,5) {$x$};
		\node[\colA] (24a) at (-3,5) {$x$};

		\node[\colA] (23c) at (-3,7) {$x$};

		\node[\colA] (31b) at (0,9) {$y$};
		\node[\colA] (32b) at (-1,9) {$y$};
		\node[\colA] (33b) at (-2,9) {$y$};
		\node[\colA] (34b) at (-3,9) {$x$};

		\node[\colA] (41a) at (0,11) {$y$};
		\node[\colA] (42a) at (-1,11) {$x$};
		\node[\colA] (43a) at (-2,11) {$x$};
		\node[\colA] (44a) at (-3,11) {$x$};

		\node[\colA] (41b) at (0,13) {$y$};
		\node[\colA] (43b) at (-2,13) {$x$};

		\node[\colA] (51) at (0,15) {$y$};
		\node[\colA] (52) at (-1,15) {$y$};
		\node[\colA] (53) at (-2,15) {$x$};
		\node[\colA] (54) at (-3,15) {$x$};

		\node[\colA] (61) at (0,17) {$y$};
		\node[\colA] (62) at (-1,17) {$x$};
		\node[\colA] (63) at (-2,17) {$x$};
		\node[\colA] (64) at (-3,17) {$x$};

		\node[\colA] (71) at (0,19) {$y$};
		\node[\colA] (73) at (-2,19) {$x$};

		\path[->]
		(11b) edge (21a)
		(12b) edge (22a)
		(13b) edge (23a)
		(14b) edge (24a)
		(12b) edge (21a)
		(14b) edge (23a)

		(31b) edge  (41a)
		(32b) edge  (42a)
		(33b) edge  (43a)
		(34b) edge  (44a)
		(33b) edge  (42a)
		(34b) edge  (43a)

		(51) edge  (61)
		(52) edge  (62)
		(53) edge  (63)
		(54) edge  (64)
		;
		
		\path
		(14a) edge[decorate,zigzagedge] (14b)
		(13a) edge[decorate,zigzagedge] (12b)
		(13a) edge[decorate,zigzagedge] (13b)
		
		(23a) edge[decorate,zigzagedge] (34b)
		(21a) edge[decorate,zigzagedge] (31b)
		(21a) edge[decorate,zigzagedge] (32b)
		(24a) edge[decorate,zigzagedge] (23c)
		
		(41a) edge[decorate,zigzagedge] (41b)
		(42a) edge[decorate,zigzagedge] (43b)
		(43a) edge[decorate,zigzagedge] (43b)

        (43b) edge[decorate,zigzagedge] (54)
        (43b) edge[decorate,zigzagedge] (53)
        (41b) edge[decorate,zigzagedge] (52)
        (41b) edge[decorate,zigzagedge] (51)
        
        (64) edge[decorate,zigzagedge] (73)
        (62) edge[decorate,zigzagedge] (73)
        (61) edge[decorate,zigzagedge] (71)
        
		;
\end{tikzpicture}}

\newcommand{\figureCaterpillar}{
\begin{tikzpicture}[scale=0.55]
\node[] at (1,0) {$\dots$};
\node[] at (7,0) {$\dots$};
\node[] at (13,0) {$\dots$};
\node[] at (19,0) {$\dots$};
\node[var-b, label={below:$v_{i_c}$}] (b) at (3,0) {};
\node[var-b, label={below:$v_{i_c+1}$}] (c) at (5,0) {};

\node[var-b, label={below:$v_{i_{c+1}}$}] (d) at (9,0) {};
\node[var-b, label={below:$v_{i_{c+1}+1}$}] (e) at (11,0) {};

\node[var-b, label={below:$v_{i_{c+2}}$}] (f) at (15,0) {};
\node[var-b, label={below:$v_{i_{c+2}+1}$}] (g) at (17,0) {};
\path[->,>=stealth']
(b) edge[very thick] (c)
(d) edge[very thick] (e)
(f) edge[very thick] (g)
;
\end{tikzpicture}}

\newcommand{\figureBraid}{
\begin{tikzpicture}[scale=0.55]

        \node at (0,1.7) {$p_c$};
\node (1) at (0,0) {$\begin{array}{cc}
		x&x  \\
		y&x\\
		y&y
	\end{array}$};

\node at (6,1.7) {$p_{c+1}$};
\node (2) at (6,0) {$\begin{array}{cc}
		x&x  \\
		y&x\\
		y&y
	\end{array}$};
    
\node at (12,1.7) {$p_{c+2}$};
\node (3) at (12,0) {$\begin{array}{cc}
		x&x  \\
		y&x\\
		y&y
	\end{array}$};
\path 
	(-3,0) edge (1)
	(1) edge (2)
	(2) edge (3)
	(3) edge (15,0)
	;
\end{tikzpicture}}

\newcommand{\fs}[1]{f}%\operatorname{fs}_{#1}}
\newcommand{\ts}[1]{t}%\operatorname{ts}_{#1}}

\usetikzlibrary{decorations.pathmorphing}

\title{%A Natural Occurrence of 
Almost Symmetric Linear Arc Monadic Datalog and Transitive Tournaments}%Caterpillar Dualities}% Constraint Satisfaction Problems}
\author{Sebastian Meyer and Florian Starke\thanks{
Both authors have received funding from the European Research Council (Project POCOCOP, ERC Synergy Grant
101071674). Views and opinions expressed are however those of the authors only and do not necessarily reflect those of the European Union or the European Research Council Executive Agency. Neither the European Union nor the granting authority can be held responsible for them.}}
\date{Juni 2026}

\begin{document}

\maketitle

\abstract{We introduce \emph{$n$-almost symmetric} Datalog and study $n$-almost symmetric linear arc monadic Datalog. We characterize the finite relational structures whose constraint satisfaction problem is solved by this Datalog fragment as those that can be primitive positively constructed from the transitive tournament on $n+2$ vertices. We also give characterizations in terms of a certain homomorphism duality (which we call \emph{$n$-fixed unfolded caterpillar duality}) and in universal-algebraic terms (the existence of \emph{$k$-absorptive operations} and of operations forming an \emph{elevator chain of length $n+1$}). This article generalizes the results from Bodirsky and Starke about symmetric linear arc monadic Datalog.}
\vspace{5mm}

\href{https://www.youtube.com/watch?v=CyuRAXDPPhI}{Here} you can find a video of a talk in the Panglobal Algebra and Logic Seminar about this article.
\section{Introduction}

Datalog fragments have been used a lot to study constraint satisfaction problems (CSPs) that are solvable in polynomial time. Well known fragments are \emph{linear Datalog}, \emph{symmetric linear Datalog}, and \emph{arc monadic Datalog}, better known as \emph{arc consistency}. In 2011 Carvalho, Dalmau, and Krokhin combined two of those restrictions and characterized the CSPs solved by linear arc monadic (lam) Datalog~\cite{lattice-ops}. Recently, Bodirsky and Starke used these results to study the CSPs solved by symmetric linear arc monadic (slam) Datalog
~\cite{bodirskyStarke2024symmetriclineararcmonadicUnfoldedCaterpillar}. One important tool when studying CSPs are \emph{primitive positive (pp-)constructions}. They give rise to a quasi order on finite relational structures which respects the computational complexity of the corresponding CSPs. The characterizations state that the CSP of a finite relational structure is solved by lam Datalog if and only if the structure can be pp-constructed by $\stCon$ and it is solved by slam Datalog if and only if the structure is pp-constructed by $\T_2$, the transitive tournament on two elements. It is well known that the transitive tournaments ordered by pp-constructability form the following chain
\[\stCon<\dots<\T_4<\T_3<\T_2<\T_1.\]
In this article we introduce $n$-almost symmetric Datalog and study $n$-almost symmetric linear arc monadic (s$_n$lam) Datalog. We show that the CSP of a finite relational structure is solved by s$_n$lam Datalog if and only if the structure can be pp-constructed by $\T_{n+2}$, generalizing the results from Bodirsky and Starke about slam Datalog.

\section{Preliminaries}

We write $[n]$ for the set $\{1,\dots,n\}$ and $[m,n]$ for the set $\{m,m+1,\dots,n\}$. 
We say that a tuple $t \in A^k$, for $k \in {\mathbb N}$, is \emph{injective} if $t$ is injective when viewed as a function from $[k]$ to $A$. 

\subsection{Structures, Homomorphisms, and CSPs}\label{sec:prelimsStructures}
We assume familiarity with the concepts of relational structures and first-order formulas from mathematical logic, as introduced for instance in~\cite{Hodges}. All structures we consider will have a finite domain and a relational signature. Let $\tau$ be a relational signature and $\bB$ be a $\tau$-structure. We say that $\bB$ is \emph{injective} if for all relationsymbols $R$ of $\tau$ all tuples in $R^{\bB}$ are injective. The structure $\bB$ is \emph{connected} if for all $b,b'\in B$ there are relationsymbols $R_1,\dots,R_n\in\tau$ and tuples $t_1\in R_1^{\bB},\dots,t_n\in R_n^{\bB}$ such that $b$ is an element of $t_1$, $b'$ is an element of $t_n$, and $t_i$ and $t_{i+1}$ have at least one element in common for all $i\in[n-1]$.  Let $\bA$ be another $\tau$-structure. A \emph{homomorphism} from $\bA\to\bB$ is a map $h\colon A\to B$ such that for all $R\in\tau$ we have $h(R^{\bA})\subseteq R^{\bB}$. The homomorphism $h$ is an \emph{embedding} if it is injective and $h(R^{\bA})= R^{\bB}$ for all $R$. An \emph{endomorphism} is a homomorphism from a structure to itself. We write $\bA\to\bB$ if there exists a homomorphism from $\bA$ to $\bB$ and $\bA\not\to\bB$ if there does not exist one. 
For $n\geq1$ a \emph{polymorphism of $\bB$ of arity $n$} is a homomorphism from $\bB^n$ to $\bB$. We denote the set of all polymorphisms of $\bB$ by $\Pol(\bB)$. 
The structures $\bA$ and $\bB$ are called \emph{homomorphically equivalent} if $\bA\to\bB$ and $\bB\to\bA$. The \emph{constraint satisfaction problem of $\bB$}, denoted $\Csp(\bB)$, is the set $\{\bA\mid \bA\text{ is a finite $\tau$-structure and $\bA\to\bB$}\}$. For example let $\stCon$ be the $\{\operatorname{source},\operatorname{target},E\}$-structure  with domain $\{s,t\}$ and relations $\operatorname{source}^\stCon\coloneqq\{s\}$, $\operatorname{target}^\stCon\coloneqq\{t\}$, and $E^\stCon\coloneqq\{(s,s),(s,t),(t,t)\}$. Then $\Csp(\stCon)$ consists of all finite directed graphs that have some vertices labeled with target and some vertices labeled with source such that there is no directed path from a source vertex to a target vertex. Hence, $\Csp(\stCon)$ encodes the complement of the directed reachability problem. 

The structure $\bB$ is a \emph{core} if all its endomorphisms are embeddings. It is well known that for each finite structure $\bB$ there is a (up to isomorphism) unique structure $\bC$ that is a core and is homomorphically equivalent to $\bB$, $\bC$ is called \emph{the} core of $\bB$.
\subsection{Primitive Positive Formulas}

A $\tau$-formula is called \emph{primitive positive (pp)} if it only uses atomic formulas ($x=y$ is allowed), conjunction, and existential quantification. There is an easy way to translate between finite $\tau$-structures and pp-$\tau$-formulas. Let $\bB$ be a $\tau$-structure, the \emph{canonical conjunctive query of $\bB$} is the pp-formula with variables $B$ that is the conjunct of all formulas $R(b_1,\dots,b_m)$ for $R\in\tau$ and $(b_1,\dots,b_m)\in R^{\bB}$. 
For a pp-$\tau$-formula $\phi$ let $\phi'$ be an equivalent pp-$\tau$-formula that is of the form $\exists x_1,\dots,x_n\, \bigwedge\text{atomic formula}$. Construct a $\tau$-structure in the following way. Take all variables (free and quantified) of $\phi'$ as the domain, a tuple $(y_1,\dots,y_m)$ is in relation $R$ if $\phi'$ contains the conjunct $R(y_1,\dots,y_m)$. Finally, identify all elements $x,y$ for which $\phi'$ contains the conjunct $x=y$ or $y=x$. The resulting structure is called the \emph{canonical database of $\phi$} and unique up to isomorphism. 

%\subsection{Primitive Positive Constructions}
%not needed? just mention as a characterization and in main thm
\subsection{Minions and Minor Conditions}
Let $\bB$ be $\tau$-structure. A \emph{clone} is a set of operations, i.e., maps from $B^n$ to $B$, that contains all projections and is closed under composition. Note that $\Pol(\bB)$ is a clone. 
%introduce idempotent?
Let $\bA$ be another $\tau$-structure. A \emph{polymorphism of $(\bA,\bB)$ of arity $n$} is a homomorphism from $\bA^n$ to $\bB$.  By $\Pol(\bA,\bB)$ we denote the set of all polymorphisms of $(\bA,\bB)$. Note that $\Pol(\bB,\bB)=\Pol(\bB)$. Let $\sigma\colon[n]\to[m]$ and $t=(a_1,\dots,a_m)\in A^m$, then $t_\sigma\coloneqq(a_{\sigma(1)},\dots,a_{\sigma(n)})\in A^n$. For a map $f\colon A^n\to B$ the \emph{minor of $f$ given by $\sigma$} is the map
$f_\sigma\colon A^m\to B, t\mapsto f(t_\sigma)$. A \emph{Minion} is a set of functions from $A^n$ to $B$ that  is closed under taking minors. Note that $\Pol(\bA,\bB)$ is a minion.   
Let $\mathcal M$ and $\mathcal N$ be minions. A map $\xi\colon \mathcal M\to\mathcal N$ is a  \emph{minion homomorphism} if $\xi$ preserves the arity of the operations and it preserves minors, i.e., for all $f\in\mathcal M$ of arity $n$ and all $\sigma\colon[n]\to[m]$ we have $\xi(f_\sigma)=(\xi(f))_\sigma$.
In this article we will be only concerned with minions of the form $\Pol(\bB)$. However, seeing that these clones are special minions makes it more plausible that we care about minion homomorphism instead of clone homomorphisms. 

Let $\tau$ be a set of function symbols each equipped with an arity. A \emph{minor condition} is a finite set of \emph{height-one identities}, i.e., expressions of the form \[f(x_1,\dots,x_n)\approx g(y_1,\dots,y_m)\]
where $f,g\in\tau$ of arities $n$ and $m$, respectively, and $x_1,\dots,x_n,y_1,\dots,y_m$ are not necessarily distinct variables. A minion $\mathcal M$ \emph{satisfies} a minor condition $\Sigma$, denoted $\mathcal M\models\Sigma$, if there is an assignment $\tau\to\mathcal M,f\mapsto f^{\mathcal M}$ that preserves the arity and for each height-one identity $f(x_1,\dots,x_n)\approx g(y_1,\dots,y_m)$ in $\Sigma$ and for each assignment $a\colon\{x_1,\dots,x_n,y_1,\dots,y_m\}\to B$ we have 
\[f^{\mathcal M}(s(x_1),\dots,s(x_n))=g^{\mathcal M}(s(y_1),\dots,s(y_m)).\]
Let $\Sigma$ and $\Sigma'$ be minor conditions. We say that $\Sigma$ \emph{implies (for finite clones)} $\Sigma'$ if every clone over a finite domain that satisfies $\Sigma$ also satisfies $\Sigma'$. Similarly we define implication for minions.
It is easy to see that minion homomorphisms compose to minion homomorphisms, hence we can quasi-order finite structures via their polymorphism clones, i.e., $\bA\leq\bB$ if there is a minion homomorphism from $\Pol(\bA)$ to $\Pol(\bB)$. This poset has been extensively studied~\cite{maximal-digraphs,vucajBodirskytwoElementPPPoset,wonderland,ZhukFVConjecture}.  
There is the following well known connection between minion homomorphisms, minor conditions, and primitive positive constructions (we will not introduce those here, for a definition see for example~\cite{starkeThesisDigraphsmoduloprimitivepositive}).

\begin{theorem}\label{thm:PPisMinorIsMinionHom}[Theorem 4.12 in~\cite{PCSP}]
	Let $\bA$ and $\bB$ be finite $\tau$-structures. Then the following are equivalent.
	\begin{enumerate}
		\item There is a minion homomorphism from $\Pol(\bB)$ to $\Pol(\bA)$.
		\item Every minor condition that holds in $\Pol(\bB)$ also holds in $\Pol(\bA)$.
		\item There is a pp-construction of $\bA$ in $\bB$.
	\end{enumerate}
\end{theorem}

The following family of minor conditions is central to the article.
\begin{definition}\label{def:hm}
For $n \geq 1$, the \emph{quasi Hagemann-Mitschke condition of length $n$}, called $\HM n$, consists of the identities %a sequence of ternary operations  $p_1,\dots,p_{n}$ on $D$ that satisfy
\begin{align*}
p_0(x) & \approx p_1(x,y,y) \\
p_{i}(x,x,y) & \approx p_{i+1}(x,y,y)  \tag*{$\text{ for all } i \in \{1,\dots,n-1\}$} \\
p_n(x,x,y) & \approx p_{n+1}(y).
\end{align*}
For $n=0$ we define $\HM0$ to be the condition consisting of the identity $p_0(x)\approx p_{1}(y)$.
%The respective height one condition is abbreviated by $\HM(n)$.
\end{definition}
Note that the polymorphism clone of a structure $\bB$ satisfies $\HM0$ if and only if $\bB$ has a constant polymorphism. 
In Figure~\ref{fig:HM3} there is a visualization of $\HM3$ that makes it clear why $p_0(x)\approx p_{1}(y)$ is a natural choice for $\HM0$.

\begin{figure}
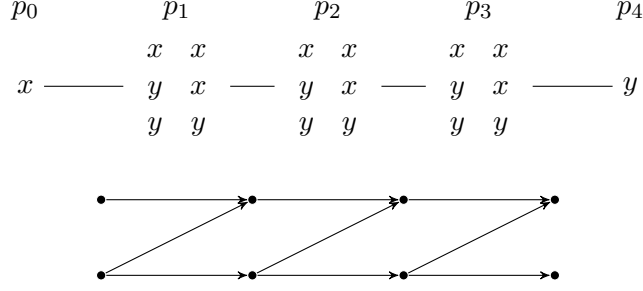

    \centering
    \HMThree
    \caption{A visualization of $\HM3$ (top) and a 3-braid (bottom), which is another way to visualize $\HM3$.}
    \label{fig:HM3}
\end{figure}

\subsection{Indicator Structures}
In this section we present a well known tool to test whether the polymorphism clone of a structure satisfies a minor condition. 
\begin{definition}\label{def:indicatorStructure}
Let $\bB$ be a relational $\tau$-structure and let $\Sigma$ be a minor condition with a single function symbol $f$ of arity $m$. 
Let $\sim$ be the smallest equivalence relation on $B^m$ such that 
 $a \sim b$ if $\Sigma$ contains $f(x_1,\dots,x_m)\approx f(y_1,\dots,y_m)$
and there is a map $s \colon \{x_1,\dots,x_m,y_1,\dots,y_m\} \to B$ with 
$a = (s(x_1),\dots,s(x_m))$ and $b = (s(y_1),\dots,s(y_m))$. Note that if $a\sim b$, then $\Sigma$ implies $f(a)\approx f(b)$, here the entries of $a$ and $b$ are viewed as variables.
Then the \emph{indicator structure of $\Sigma$ with respect to $\bB$} is the $\tau$-structure $\bB^m/_{\sim}$, where the tuple $([a_1]_\sim,\dots,[a_n]_\sim)$ is relation $R^{\bB^m/_{\sim}}$ if there are $a'_1\in[a_1]_\sim,\dots,a'_n\in[a_n]_\sim$ such that $(a'_1,\dots,a'_n)\in R^{\bB^m}$. 
%The domain of $\bB^m/_{\sim}$ are the equivalence classes of $\sim$. 
\end{definition}

The following is straightforward from the definitions. 

\begin{lemma}\label{lem:ind} 
Let $\bB$ be a structure and $\Sigma$ be a minor condition with a single function symbol. Then 
$\bB$ has a polymorphism satisfying $\Sigma$ if and only if the indicator structure of $\Sigma$ with respect to $\bB$ has a homomorphism to $\bB$.
\end{lemma}

It is well known that for clones over a finite domain the restriction to minor conditions with a single function symbol is inessential, for a proof see Lemma 2.21 in \cite{bodirskyStarke2024symmetriclineararcmonadicUnfoldedCaterpillar}.
\begin{lemma}\label{lem:WlogSingleFunctionSymbol}
	 Let $\Sigma$ be a minor condition. Then there exists a minor condition $\Sigma'$ with a single function symbol such that 
	a clone over a finite domain satisfies $\Sigma$ if and only if it satisfies $\Sigma'$. 
\end{lemma}

\subsection{Caterpillars }
The notion of caterpillars was generalized from graphs to relational structures by Carvalho, Dalmau, and Krokhin in \cite{lattice-ops}. 
We will use an equivalent definition here that makes it easier to transform caterpillars.

\begin{definition}
    An injective $\tau$-structure $\bC$ is a \emph{caterpillar} if there is a tuple $(a_0,\dots,a_n)$ of distinct elements of $\bC$ and a pp-formula $(\phi\wedge\psi)(a_0,\dots,a_n)$ such that
    \begin{itemize}
        \item $\bC$ is isomorphic to the canonical database of $\phi\wedge\psi$,
        \item the formula $\phi$ is of the form 
        \[\phi_0(a_0,a_1)\wedge\dots\wedge\phi_{n-1}(a_{n-1},a_n),\]
        where each $\phi_i(a_i,a_{i+1})$ is obtained from an atomic formula, where $a_i$ and $a_{i+1}$ occur, by existentially quantifying all variables except for $a_i$ and $a_{i+1}$,
        \item the formula $\psi$ is of the form 
        \[\psi_0(a_0)\wedge\dots\wedge\psi_n(a_n),\]
        where each $\phi_i(a_i)$ is a (possibly empty) conjunct of formulas obtained from an atomic formula, where $a_i$ occurs, by existentially quantifying all variables except for $a_i$, and
        \item if $\bC$ is not an isolated point, then $\psi_0$ and $\psi_n$ are not the empty conjunct.
    \end{itemize}

    The elements $a_0,\dots,a_n$ are called \emph{joints} of $\bC$, the tuples $(a_k,a_{k+1})$ are called \emph{vertebra}, and the tuple $(a_0,\dots,a_n)$ is called the \emph{spine} of $\bC$. A vertebra $(a_k,a_{k+1})$ \emph{lies between} $a_i$ and $a_{j}$ if $i\leq k$ and $k+1\leq j$. The conjuncts of $\psi_i(a_i)$ are called \emph{hairs} of $a_i$. If $\psi_i$ is the empty conjunct we say that $a_i$ has \emph{no hairs}. The formula $\phi\wedge\psi$ is the \emph{conjunctive query representation} of $\bC$.
\end{definition}

\begin{example}
	The structure $\bC$ with domain $\{a_0,a_1,a_2,b_0,b_1,b_2,c\}$ and relations $R^{\bC}=\{(a_1,a_0,c)\},E^{\bC}=\{(a_1,a_2),(a_0,b_0),(b_2,a_2)\},S^{\bC}=\{(a_0,b_1)\}$ is a caterpillar with spine $(a_0,a_1,a_2)$. Its conjunctive query representation consists of the formulas
	\begin{align*}
	\phi_0(a_0,a_1)&\coloneqq \exists x\,R(a_1,a_0,x), &
	\phi_1(a_1,a_2)&\coloneqq E(a_1,a_2),\\
	\psi_0(a_0)&\coloneqq \exists x\,E(a_0,x)\wedge \exists x\,S(a_0,x),& 
	%\item $\psi_1(a_1)$ is the empty conjunct, and
	\psi_2(a_2)&\coloneqq \exists x\,E(x,a_2),\text {and}
	\end{align*}
	$\psi_1(a_1)$ is the empty conjunct.
	The caterpillar $\bC$ can be seen as the following labeled directed path 
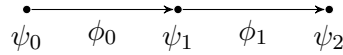
\begin{center}
	\begin{tikzpicture}
		\node[bullet,label=below:$\psi_0$] (0) at (0,0) {};
		\node[bullet,label=below:$\psi_1$] (1) at (2,0) {};
		\node[bullet,label=below:$\psi_2$] (2) at (4,0) {};
		
		\path[->,>=stealth']
		(0) edge node[below] {$\phi_0$} (1)
		(1) edge node[below] {$\phi_1$} (2)
		;
	\end{tikzpicture}
\end{center}	
Note that the orientation of the spine is not unique, we could also take $(a_2,a_1,a_0)$ as the spine of $\bC$. These are the only two possibilities to choose a spine, hence using the definite article is still justified.
	%connected by the vertebrae $((a_2,a_1,c),R,2,1)$ and $((a_2,a_3),E,1,2)$. It has hairs $((a_1,b_1),E,1)$, $((a_1,b_1),S,1)$, and $((a_3,b_2),E,1)$.
	%non examples directed cycle, disconnected
\end{example}

%example + discussion uniqness of formula/spine (if caterpillar is large enough and we assume head and tail have hairs)

%which spine we take will be clear from the context (reversed one or not reversed)

\begin{definition}\label{def:unfolding}
    Let $\bC$ be a caterpillar with spine $(a_0,\dots,a_n)$ and conjunctive query representation $(\phi\wedge\psi)(a_0,\dots,a_n)$. Let  $i,j\in[0,n]$ with $i<j$ and define
        \begin{itemize}
            \item $\phi_{\text{head}}(a_0,\dots,a_{i})\coloneqq \phi_0(a_0,a_1)\wedge\dots\wedge\phi_{i-1}(a_{i-1},a_i)$,
            \item $\phi_{\text{body}}(a_i,\dots,a_{j})\coloneqq \phi_i(a_i,a_{i+1})\wedge\dots\wedge\phi_{j-1}(a_{j-1},a_j)$,
            \item $\phi_{\text{tail}}(a_j,\dots,a_{n})\coloneqq \phi_j(a_j,a_{j+1})\wedge\dots\wedge\phi_{n-1}(a_{n-1},a_n)$,
            \item $\psi_{\text{head}}(a_0,\dots,a_{i})\coloneqq \psi_0(a_0)\wedge\dots\wedge\psi_{i}(a_i)$,
            \item $\psi_{\text{body}}(a_{i+1},\dots,a_{j-1})\coloneqq \psi_{i+1}(a_{i+1})\wedge\dots\wedge\psi_{j-1}(a_{j-1})$, and
            \item $\psi_{\text{tail}}(a_j,\dots,a_{n})\coloneqq \psi_j(a_j)\wedge\dots\wedge\psi_{n}(a_n)$.
        \end{itemize}

The \emph{$(a_i,a_j)$-unfolding} of $\bC$ is the canonical database $\bC'$ of the pp-formula $\phi'\wedge\psi'$ with
      \begin{align*}
       \phi'\coloneqq \phi_{\text{head}}(a_0,\dots,a_{i})\wedge{}
      &\phi_{\text{body}}(a_i,a'_{i+1}\dots,a'_{j-1},a'_{j})\wedge{}\\
      &\phi_{\text{body}}(a''_i,a_{i+1}\dots,a_{j-1},a'_{j})\wedge{}\\
      &\phi_{\text{body}}(a''_i,a''_{i+1}\dots,a''_{j-1},a_{j})\wedge
      \phi_{\text{tail}}(a_j,\dots,a_{n})
\intertext{and}
       \psi'\coloneqq \psi_{\text{head}}(a_1,\dots,a_{i})\wedge{}
      &\psi_{\text{body}}(a'_{i+1}\dots,a'_{j-1})\wedge{}\\
      &\psi_{\text{body}}(a_{i+1}\dots,a_{j-1})\wedge{}\\
      &\psi_{\text{body}}(a''_{i+1}\dots,a''_{j-1})\wedge
      \psi_{\text{tail}}(a_j,\dots,a_{n})
      \end{align*}
      It is easy to see that $\bC'$ is a caterpillar with spine 
      \[(a_1,\dots,a_{i},a'_{i+1},\dots,a'_{j-1},a'_{j},a_{j-1},\dots,a_{i+1},a''_i,a''_{i+1},\dots,a''_{j-1},a_{j},\dots,a_n).\]
      
Observe that $a'_j$ and $a_i''$ have no hairs and that for $k\in\{0,\dots,i-1,j,\dots,n-1\}$ the caterpillar $\bC'$ also has the vertebra $(a_k,a_{k+1})$,
the $(a_i,a_j)$-unfolding of $\bC$ is said to \emph{fix} the vertebra $(a_k,a_{k+1})$.
We will assume that for an element $c$ that is in the head or the tail of $\bC$,  the corresponding element in $\bC'$ is also called $c$ (the elements $a_i$, $a_j$, and the elements in their hairs also belong to the body and tail, respectively). For any element $c$ in the body of $\bC$ there are three corresponding elements in $\bC'$, our convention will be that the first one is called $c'$, the second $c$, and the third $c''$. Note that this convention extends the naming convention we used for the joints. This way we obtain a surjective homomorphism $h\colon \bC'\to\bC$ that is the identity on the tail and the head and sends the elements $c,c'$ and $c''$ from the body to $c$.

An \emph{unfolding} $\bC'$ of $\bC$ that fixes $v_{1},\dots, v_{m}$ is a structure that is obtained from $\bC$ by repeatedly applying $(a,b)$-unfoldings that fix $v_{1},\dots, v_{m}$. This is well defined, since whenever an $(a,b)$-unfolding fixes some vertebra $v_i$, then the unfolded caterpillar still has the vertebra $v_i$. Note that, by composing the previously constructed homomorphisms, we obtain a surjective homomorphism $h\colon\bC'\to\bC$ that is the identity on all joints in $v_{1},\dots, v_{m}$.
\end{definition}

\subsection{Datalog Fragments}
\label{sect:datalog}
For a detailed introduction, see, e.g.,~\cite{BodDalJournal}. 
Let $\tau$ and $\rho$ be finite relational signatures such that $\tau \subseteq \rho$. 
A \emph{Datalog program} $\Pi$ is a finite set of rules of the form
$$ \phi_0 \; {:}{-} \; \phi_1\wedge\dots\wedge\phi_n$$
where each $\phi_i$ is an atomic $\tau$-formula. The formula $\phi_0$ is called the \emph{head} of the rule, and $\phi_1\wedge\dots\wedge\phi_n$ is called the \emph{body} of the rule. The symbols in $\tau$ are called \emph{EDBs} (\emph{extensional database predicates}) and the other symbols from $\rho\setminus\tau$ are called \emph{IDBs} (\emph{intensional database predicates}). In the rule heads, only IDBs are allowed. There is one special IDB $G$ of arity 0, which is called the \emph{goal predicate}. 
IDBs might also appear in the rule bodies. A \emph{derivation} of $\Pi$ on a finite $\tau$-structure $\bA$ is a sequence $\bA_0\vdash_{R_0} \dots\vdash_{R_{n-1}} \bA_n$ of $\rho$-structures $\bA_0,\dots,\bA_n$ and rules $R_0,\dots,R_{n-1}$ of $\Pi$ such that 
\begin{itemize}
	\item the $\tau$-reduct of $\bA_0$ is $\bA$ and $P^{\bA_0}=\emptyset$ for all $P\in\rho\setminus\tau$ and
	\item for all $i\in[0,n-1]$ the structure $\bA_{i+1}$ is obtained from $\bA_i$ by adding a tuple $(s(x_1),\dots,s(x_k))$ to $P^{\bA_i}$, where $R_i$ is $P(x_1,\dots,x_n)\; {:}{-} \;\phi$ and $s$ is a satisfying assignment of the variables of $P(x_1,\dots,x_n)\wedge\phi$ to elements of $\bA_i$.
\end{itemize}
We say that this derivation \emph{derives the IDB $P$ on the tuple $t$} if $\bA_n$ is obtained from $\bA_{n-1}$ by adding $t$ to $P^{\bA_{n+1}}$.  
We say that this derivation \emph{derives the goal predicate (on $\bA$)} if it derives $G$ on the empty tuple $()$, i.e., if $G^{A_n}\not=\emptyset$.
%We view the set of rules as a recursive specification of the IDBs in terms of the EDBs -- for a detailed introduction, see, e.g.,~\cite{BodDalJournal}. 
A Datalog program is called 
\begin{itemize}
    \item \emph{linear} if in each rule, at most one IDB appears in the body (we then assume without loss of generality that in every rule whose body contains an IDB, the IDB is listed first). 
\item \emph{arc} if each rule involves at most one EDB. 
\item \emph{symmetric} if it is linear and for every rule $\phi_0 \; {:}{-} \; \phi_1\wedge\phi_2\wedge\dots\wedge\phi_n$ where $\phi_0$ is not the goal predicate and $\phi_1$ is an atomic formula whose relation symbol is an IDB, 
the Datalog program also contains the \emph{reversed rule}  $\phi_1 \; {:}{-} \; \phi_0\wedge\phi_2\wedge\dots\wedge\phi_n$.\footnote{Note that excluding the case that $\phi_0$ is the goal predicate is not essential, because 
we may always add its symmetric version without changing the set of structures on which the goal predicate is derived. The advantage of this version is that symmetric Datalog programs do not need to have rules with the goal predicate in their body.} 
\item \emph{$n$-almost symmetric} if it is linear and there is a quasi-order $\leq$ on the IDBs such that 
\begin{itemize}
    \item for every strict chain $P_1<\dots<P_m$ of IDBs we have $m\leq n+1$ and
    \item for every rule $\phi_0 \; {:}{-} \; \phi_1\wedge\phi_2\wedge\dots\wedge\phi_n$ where $\phi_0$ and $\phi_1$ are atomic formulas whose relation symbols are the IDBs $P_0$ and $P_1$, respectively, we have $P_0\leq P_1$, and if $P_1\leq P_0$, then the Datalog program also contains the reversed rule.
\end{itemize}
\end{itemize}
Note that any derivation of an $n$-almost symmetric Datalog program can use non-symmetric rules at most $n$ times. In particular, a 0-almost symmetric Datalog program is symmetric.

If $\bB$ is a $\tau$-structure, then 
we say that $\Csp(\bB)$ is \emph{solved} by a Datalog program $\Pi$ with EDBs $\tau$ if the following holds: the goal predicate can be derived by $\Pi$ on a finite $\tau$-structure $\bA$ if and only if there is \emph{no} homomorphism from $\bA$ to $\bB$. 
We say that a Datalog program has \emph{width $(\ell,k)$}
if all IDBs have arity at most $\ell$, and if every rule has at most $k$ variables. 
For given $(\ell,k)$ and a structure $\bB$, there exists a Datalog program $\Pi$ of width $(\ell,k)$ with the remarkable 
property that if some Datalog program of width $(\ell,k)$ solves $\Csp(\bB)$, then $\Pi$ solves $\Csp(\bB)$. 
This Datalog program is referred to as the \emph{canonical Datalog program for $\bB$ of width $(\ell,k)$}, and is constructed as follows~\cite{FederVardi}:
For every relation $R$ over $B$ of arity at most $\ell$, we introduce a new IDB. The empty relation of arity 0 plays the role of the goal predicate. 
Then $\Pi$ contains all rules 
$\phi_0 \; {:}{-} \; \phi_1\wedge\dots\wedge\phi_n$
with at most $k$ variables 
such that the formula 
$\forall \bar x (\phi_1 \wedge \dots \wedge \phi_n \Rightarrow \phi_0)$ holds in the expansion of $\bB$ by all IDBs. 
Note that if the canonical Datalog program $\Pi$ for $\bB$ does not solve $\Csp(\bB)$, then it is still true that if $\Pi$ derives the goal predicate on a finite structure $\bA$, then 
there is no homomorphism from $\bA$ to $\bB$,  (see, e.g.,~\cite{BodDalJournal}).

If $k$ is the maximal arity of the EDBs, 
we may restrict the canonical Datalog program of width $(1,k)$ to those rules with only unary IDBs and at most one EDB; in this case, we obtain the canonical arc monadic Datalog program, %\todo{cool: for this fragment, no parameters $(\ell,k)$ needed!} 
 which is also known as the \emph{arc consistency procedure}. 
Analogously, we may define the canonical \emph{linear}, %\todo{Here, we would need parameters, but let's not mention it again...} 
and the canonical \emph{symmetric} Datalog program. We may also combine these restrictions, and in particular obtain 
a definition the \emph{canonical slam Datalog program}, i.e., the canonical symmetric linear arc monadic Datalog program.

\begin{definition}
Let $\bB$ be a finite structure with a finite relational signature, the \emph{canonical s$_n$lam Datalog program for $\bB$} is obtained by taking $n+1$ many copies $\Pi_0,\dots,\Pi_n$ of the canonical slam Datalog program and adding for each rule $\phi_0 \; {:}{-} \; \phi_1\wedge\phi_2\wedge\dots\wedge\phi_n$ of the canonical lam Datalog program, where $\phi_0$ and $\phi_1$ are atomic formulas whose relation symbols are the IDBs $P_0$ and $P_1$, respectively, the rules $\phi^i_0 \; {:}{-} \; \phi^{i+1}_1\wedge\phi_2\wedge\dots\wedge\phi_n$ for $i\in[0,n-1]$, where $\phi^i_0$ is obtained from $\phi_0$ by replacing $P_0$ by its copy from $\Pi_i$ and $\phi^{i+1}_1$ is obtained from $\phi_1$ by replacing $P_1$ by its copy from $\Pi_{i+1}$. Clearly, the resulting program is $n$-almost symmetric, linear, arc, and monadic.
\end{definition}
Note that the canonical s$_0$lam Datalog program is the canonical slam Datalog program.
The following lemma can be shown analogously to the well-known fact for unrestricted canonical Datalog programs of width $(\ell,k)$ (see, e.g.,~\cite{BodDalJournal}). %\todo[inline]{Do we need this? If not, remove...}

\begin{lemma}\label{lem:can-sound}
Let $\bB$ be a finite structure with a finite relational signature, $n\in\N$,
and let $\Pi$ be the canonical s$_n$lam Datalog program for $\bB$. 
If $\bA$ is a finite structure with a homomorphism to $\bB$, then 
$\Pi$ does not derive the goal predicate on $\bA$. 
\end{lemma} 
%add notation for derivation of linear DL program\todo[inline]{MB: why  only for linear? FS: then we can define a derivation to be a sequence (thats convinient for the proof of Lemma 3.7). MB: I think it should be equally easy to give a general definition that specialises to the one you want when the program is linear.}

\subsection{Dualities}

Let $\bB$ be a $\tau$-structure and $\mathcal F$ be a set of $\tau$-structures. The tuple $(\mathcal F,\bB)$ is a \emph{duality pair} if for all structures $\bA$ with $\bA\not\to\bB$ there is a $\bF\in\mathcal F$ with $\bF\to\bA$. The set $\mathcal F$ is called an \emph{obstruction set of $\bB$} and its elements \emph{obstructions}. We can now describe structures by properties their obstruction sets. The most commonly known are \emph{finite duality}, i.e. there is a finite obstruction set, and \emph{tree duality}, i.e., there is an obstruction set consisting of trees. In 2011 Carvalho, Dalmau, and Krokhin introduced \emph{caterpillar duality}, i.e., there is an obstruction set consisting of caterpillars, and characterized when a structure does have caterpillar duality. This characterization was later extended by Bodirsky and Starke by the items (5)-(7) involving the structure $\stCon$ introduced in Section~\ref{sec:prelimsStructures}.
\begin{theorem}\label{thm:mainOld}[Theorem 16 in \cite{lattice-ops} and Remark 5.1 in~\cite{bodirskyStarke2024symmetriclineararcmonadicUnfoldedCaterpillar}]
	Let $\bB$ be a finite relational $\tau$-structure. 
	Then the following are equivalent. 
	\begin{enumerate}
		\item  $\bB$ has caterpillar duality. 
		\item $\Csp(\bB)$ can be solved by a linear arc monadic Datalog program. 
		%\item $C(\bB) \to \bB$.
		\item $\Pol(\bB)$ contains for every $k,n \geq 1$ an $k$-absorbing operation of arity  $kn$. 
		\item 
		$\bB$ is homomorphically equivalent to a structure $\bB'$ with binary polymorphisms $\sqcup$ and $\sqcap$ such that $(B',\sqcup,\sqcap)$ is a (distributive) lattice. 
		\item  Every minor condition that holds in 
		$\Pol(\stCon)$ also holds in $\Pol(\bB)$. 
		\item There is a minion homomorphism from $\Pol(\stCon)$ to $\Pol(\bB)$.
		\item There is a primitive positive construction of $\bB$ in $\stCon$. 
	\end{enumerate}
\end{theorem}
The $k$-absorbing operations of arity $kn$ from item (3) are introduced under the name $k$-ABS operations in Section~3.1 in \cite{lattice-ops}. For us it is only important that whether an operation is $k$-absorbing can be described by a minor condition. 
Bodirsky and Starke introduced \emph{unfolded caterpillar duality}, i.e., there an obstruction set consisting of caterpillars that is closed under unfoldings, and used the result about caterpillar duality to obtain similar characterizations for structures with unfolded caterpillar duality~\cite{bodirskyStarke2024symmetriclineararcmonadicUnfoldedCaterpillar}. In this article we will generalize their result.

Finally, one general observation about duality pairs that motivates why we defined caterpillars only for injective structures. Using the sparse incomarability theorem it is easy to show that we can assume without loss of generality that obstruction sets only contain  injective structures, for a proof see for example Lemma 3.7 in~\cite{bodirskyStarke2024symmetriclineararcmonadicUnfoldedCaterpillar}.
\begin{lemma}\label{lem:inj}
	Let $\bB$ be a finite structure and let $\mathcal F$ be a class of finite structures such that $(\mathcal F,\bB)$ is a duality pair. Define 
	\[\mathcal F'\coloneqq \{\bF\in\mathcal F \mid \text{$\bF$ is injective}\}.\]
	%\[\mathcal F'\coloneqq \{\bF\in\mathcal F \mid \text{all tuples in relations of $\bF$ are injective}\}.\]
	Then $(\mathcal F',\bB)$ is a duality pair.
\end{lemma}

\subsection{Transitive tournaments and directed reachability}
For $n\geq1$ define the \emph{transitive tournament on $n$ vertices} to be the $\{E\}$-structure \[\T_n\coloneqq([0,n-1],\{(i,j)\mid 0\leq i<j<n\}).\] These structures will be central to this article. We now give a selection of well known facts about these tournaments for more information see for example Section 5.1 in~\cite{starkeThesisDigraphsmoduloprimitivepositive}.
It is true that there is a minion homomorphism from $\Pol(\stCon)$ to $\Pol(\T_n)$ for all $n\geq1$ and that there is a minion homomorphism from $\Pol(\T_n)$ to $\Pol(\T_m)$ if and only if $n\geq m$. Hence, we can order these structures in the following way
\[\T_1<\T_2<\T_3<\T_4<\dots<\stCon\]
%\[\stCon>\dots >\T_4>\T_3>\T_2>\T_1\]
Furthermore, $\stCon$ is the limit of this chain of $\T_n$'s in the sense that for any structure $\bB$ with $\T_n\leq \bB$ for all $n\geq1$ we also have $\stCon\leq\bB$. 
This chain lies at the top of the poset of all finite structures ordered by minion homomorphisms. The structure $\T_1$ is the unique maximal element and $\T_2$ is its unique lower cover. Meyer and Starke showed in 2024 that $\T_2$ has infinitely many lower covers: one for each finite simple group and $\T_3$. Unfortunately, $\T_4$ is not a lower cover of $\T_3$ as we will show in Example~\ref{exa:aboveT4AndBelowT3}. 
It is well known that transitive tournaments are closely related to Hagemann-Mitschke chains in the following way. For a proof see Theorem 5.4  in~\cite{starkeThesisDigraphsmoduloprimitivepositive}.
\begin{theorem}\label{thm:TnAreHMBlocker}
	Let $\bB$ be a finite structure and $n\geq0$. Then $\bB\leq\T_{n+2}$ if and only if $\Pol(\bB)\not\models\HM n$. 
\end{theorem}
In particular, $\T_{n+2}\not\models\HM n$ and $\T_{n+2}\models\HM {n+1}$. 
For $n\geq0$ define the \emph{directed path of length $n$} to be the $\{E\}$-structure $\bP_n\coloneqq([0,n],\{(i-1,i)\mid i\in[n]\})$. It is easy to see that $(\{\bP_n\},\T_n)$ is a duality pair for all $n\geq1$. Hence $\T_n$ has finite duality and, since $\P_n$ is a caterpillar, $\T_n$ also has caterpillar duality. We can now give a typical example of a s$_n$lam Datalog program.
\begin{example}\label{exa:TnSolvedBySlamDL}
Let $\T_{n+2}'$ be the $\{E,R\}$-expansion of $\T_{n+2}$ where $R^{\T_{n+2}'}$ is the equality relation. Observe that the Datalog program consisting of the rules 
\begin{align*}
     P_1(y)\; &{:}{-} \;E(x,y) &
     P_{i+1}(y)\; &{:}{-} \;P_i(x)\wedge E(x,y) &\tag*{$i\in [n]$}\\
     P_i(y)\; &{:}{-} \;P_i(x)\wedge R(x,y) &
     P_i(x)\; &{:}{-} \;P_i(y)\wedge R(x,y) &
      &\tag*{$i\in[n+1]$}\\
     G\; &{:}{-}\; P_{n+1}(x)\wedge E(x,y)
\end{align*}
%and the reversed rules of $P_i(y)\; {:}{-} \;P_i(x)\wedge R(x,y)$ 
solves $\Csp(\T_{n+2}')$. Note that this program is  $n$-almost symmetric, linear, arc, and monadic.     
\end{example}

\section{Caterpillars and Elevators}
We generalize the notion of unfolded caterpillar duality.
For $n\geq0$ a set $\mathcal F$ of finite structures is \emph{closed under $n$-fixed unfoldings (of caterpillars)} if for every caterpillar $\bC$ in $\mathcal F$, there are vertebrae  $v_1,\dots,v_n$ such that any unfolding of $\bC$ that fixes $v_1,\dots,v_n$ is contained in $\mathcal F$. The set $\mathcal F$ is closed under \emph{unfoldings (of caterpillars)} if it is closed under 0-fixed unfoldings.
A structure $\bB$ has \emph{$n$-fixed unfolded caterpillar duality} if there exists a set $\mathcal F$ of caterpillars such that 
\begin{itemize}
	\item $(\mathcal F,\bB)$ is a duality pair and
	\item $\mathcal F$ is closed under $n$-fixed unfoldings of caterpillars.
\end{itemize}
Note that unfolded caterpillar duality is 0-fixed unfolded caterpillar duality.

\begin{observation}
	For all $n\geq0$ the path $\P_{n+2}$ is a caterpillar with spine $(1,\dots,n)$. Hence, $\P_{n+2}$ has $n$ vertebrae. Since $(\{\P_{n+2}\},\T_{n+2})$ is a duality pair, $\T_{n+2}$ has $n$-fixed unfolded caterpillar duality, as we can fix all vertebra in $\P_{n+2}$.
\end{observation}

Clearly, we can generalize this observation to the following lemma.
\begin{lemma}\label{lem:finiteCaterpillarDuality}
	 If $\bB$ is a finite structure and $\mathcal F$ is a finite set of caterpillars such that $(\mathcal F,\B)$ is a duality pair, then $\bB$ has $n$-fixed unfolded caterpillar duality, where 
	$n$ is the maximal number of vertebrae in a caterpillar in $\mathcal F$.
\end{lemma}

Let $n\in\N$, $\bC$ be a caterpillar with spine $(a_0,\dots,a_m)$, and  $0\leq i_1< j_1\leq\dots\leq i_n< j_n\leq m$. For $k\in[n]$ let $\bC_k$ denote the $(a_{i_k},a_{j_k})$-unfolding of $\bC$. Then $\bC$ is an \emph{$n$-folding} of $(\bC_1,\dots,\bC_n)$. A set of structures $\mathcal B$ is \emph{closed under $n$-foldings (of caterpillars)} if for all caterpillars $\bC_1,\dots,\bC_n\in\mathcal B$ we have that all $n$-foldings of $(\bC_1,\dots,\bC_n)$ are in $\mathcal B$.

%\section{Elevator Condition}
In his dissertation Starke introduced the elevator condition~\cite{starkeThesisDigraphsmoduloprimitivepositive}. Here we give an updated version of this definition. The identities are different but they capture the same idea.

\begin{definition}\label{def:elevator}
For $n \geq 1$, the \emph{elevator condition of length $n$}, called $\Ele n$, consists of the identities 
%\begin{align*}
%e_0(x_1,\dots,x_n) & \approx e_1(x_1,y,y,x_2,\dots,x_n) \\
%e_{i}(x_1,\dots,x_{i-1},y,y,x_i,\dots,x_n) & \approx e_{i+1}(x_1,\dots,x_i,z,z,x_{i+1},\dots,x_n) && \text{ for all } i \in [n-1] \\
%e_n(x_1,\dots,x_{n-1},y,y,x_n) & \approx e_{n+1}(x_1,\dots,x_n).
%\end{align*}
\begin{align*}
e_0(x_1,\dots,x_n)  \approx e_1&(x_1,z,z,x_2,\dots,x_n)
\intertext{for all $i\in[n-1]$}
e_{i}&(x_1,\dots,x_{i-1},y,y,x_i,x_{i+1},x_{i+2},\dots,x_n)\\
\approx{} e_{i+1}&(x_1,\dots,x_{i-1},x_i,x_{i+1},z,z,x_{i+2},\dots,x_n)
\intertext{and}
e_n&(x_1,\dots,x_{n-1},y,y,x_n)  \approx e_{n+1}(x_1,\dots,x_n).
\end{align*}
\end{definition}
See Figure~\ref{fig:Ele3} for a visualization of $\Ele3$. A set of operations that satisfies $\Ele n$ is called an \emph{elevator chain of length $n$}. 
\begin{figure}
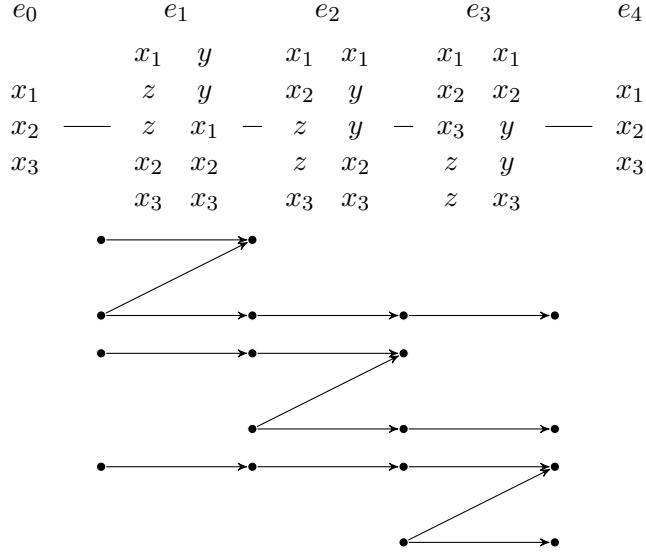

    \centering
    \EleThree
    \caption{A visualization of $\Ele3$ (top) and a 3-staircase (bottom), which is another way to visualize $\Ele3$.}
    \label{fig:Ele3}
\end{figure}
Note that $\Ele 1$ is equivalent to $\HM 1$. In general, we have the following connection.
\begin{lemma}\label{lem:EleImpliesHM}
	Let $n\geq1$. Then $\Ele n$ implies (for minions) $\HM n$.
\end{lemma}
\begin{proof}
	Let $\mathcal M$ be a minion and let $e_0,\dots,e_{n+1}\in\mathcal M$ form an elevator chain of length $n$. Define functions $p_0,\dots,p_{n+1}$ in the following way
	\begin{align*}
		p_0(a)&\coloneqq e_0(a,\dots,a)\\
p_i(a,b,c)&\coloneqq e_i(c,\dots,c,a,b,c,a,\dots,a)\text{ with $b$ at position $i+1$}		\\
p_{n+1}(c)&\coloneqq e_{n+1}(c,\dots,c)
	\end{align*}
for all $a,b,c$. Note that $p_0,\dots,p_{n+1}$ are minors of $e_0,\dots,e_{n+1}$, respectively. Therefore, they are elements of $\mathcal M$. Furthermore, they form a Hagemann-Mitschke chain. Hence, $\mathcal M\models\HM n$.
\end{proof}

Observe that $\Ele n$ implies
\begin{align*}
e_0(x_1,\dots,x_n)  \approx{} &e_1(x_1,x_1,x_1,x_2,\dots,x_n)\approx\cdots\\
\dots\approx{} &e_n(x_1,\dots,x_{n-1},x_n,x_n,x_n) 
\approx e_{n+1}(x_1,x_2,\dots,x_n).
\end{align*}
Hence, $e_0=e_{n+1}$ in every elevator chain of length $n$.

\section{Main Result}
Now we can formulate our main result.

\begin{theorem}\label{thm:main}    
Let $n\geq0$ and let $\bB$ be a structure with a finite domain and a finite relational signature. 
Then the following are equivalent. 
    \begin{enumerate}
        \item \label{maltsev}
        $\Pol(\bB)$ contains operations that form an elevator chain of length $n+1$ and $k$-absorptive operations of arity $\ell k$, for all $\ell,k \geq 1$.
        \item \label{closedUnderFoldings}
        $\bB$ has caterpillar duality and $\Csp(\bB)$ is closed under $(n+1)$-foldings of caterpillars. 
        \item 
        \label{can-sym-lin-arc}
        The canonical s$_n$lam  Datalog program for $\bB$ solves $\Csp(\bB)$.
        \item \label{sym-lin-arc}
        Some s$_n$lam Datalog program solves $\Csp(\bB)$.
        \item 
        \label{caterpillarAndUnfolding}
        $\bB$ has caterpillar duality and $\overline{\Csp(\bB)}$ is closed under $n$-fixed unfoldings of caterpillars. 
        \item 
        \label{caterpillar}
        $\bB$ has $n$-fixed unfolded caterpillar duality. 
        %\item 
        %\label{finite-caterpillar}
        %$\bB$ has finitely based $n$-fixed unfolded caterpillar duality. 
        %there is a set of caterpillars $\mathcal F$ that is closed under $n$-fixed unfoldings such that $(\mathcal F,\bB)$ is a duality pair $\bB$ has $n$-fixed unfolded caterpillar duality and there are caterpillars $\bC_1,\dots,\bC_m\in\mathcal F$ for which for every caterpillar $\bC$ in $\mathcal F$ there is an $i\in[m]$ such that $\bC$ is an unfolding of $\bC_i$.
        \item \label{non-degenerate-minor}
        If $\Pol(\bB)$ does not satisfy a minor condition $\Sigma$,
        then $\Sigma$ implies (for finite clones) $\HM n$. 
        \item \label{minor-p2} Every minor condition that holds in $\Pol(\bT_{n+2})$ also holds in $\Pol(\bB)$. 
        \item \label{minion-hom} There is a minion homomorphism from $\Pol(\bT_{n+2})$ to $\Pol(\bB)$. 
        \item \label{pp-p2} There is a primitive positive construction of 
        $\bB$ in $\bT_{n+2}$.
        \item \label{lattice} $\bB$ is homomorphically equivalent to a structure $\bB'$ such that $\Pol(\bB')$ contains operations that form an elevator chain of length $n$ and operations $\sqcup$ and $\sqcap$ such that $(B',\sqcup,\sqcap)$ forms a (distributive) lattice.
    %\item \label{minion-core} (Libor Barto, personal communication) The minion core of $\Pol(\bB)$ equals $\Pol(\bC_2,\bB_2)$. 
    % THIS ITEM IS ALREADY COVERED BY A COMMENT THAT FOLLOWS THE THEOREM. 
%            \item
%        $\Pol(\bB)$ contains totally symmetric operations of all arities and $\operatorname{GM}(n)$ operations
%        of arity $n$, for all odd $n \geq 2$. (Theorem 3.12 in Submaximal Clones over a Three-Element set)
    \end{enumerate}
    %Moreover, if one of these items holds, {and $\bB$ is a core,} then there exists a structure $\bB'$ with a binary relational signature such that  $\Pol(\bB') = \Pol(\bB)$, and all the statements hold for $\bB'$ in place of $\bB$ as well. 
\end{theorem}
For $n=0$ this is Theorem~3.1 from~\cite{bodirskyStarke2024symmetriclineararcmonadicUnfoldedCaterpillar}.
The proof of this theorem will stretch out over the rest of this article. 
We will prove $(\ref{maltsev})
\Rightarrow(\ref{closedUnderFoldings})
\Rightarrow(\ref{caterpillarAndUnfolding})
\Rightarrow(\ref{caterpillar})
\Rightarrow(\ref{non-degenerate-minor})
\Rightarrow(\ref{minor-p2})
\Rightarrow(\ref{maltsev})$ and $(\ref{caterpillarAndUnfolding})
\Rightarrow(\ref{can-sym-lin-arc})
\Rightarrow(\ref{sym-lin-arc})
\Rightarrow(\ref{caterpillarAndUnfolding})$. The equivalence of $(\ref{minor-p2})-(\ref{pp-p2})$ follows from general results. Finally, we will show $(\ref{maltsev})\Leftrightarrow(\ref{lattice})$.

\section{Folding and Unfolding Caterpillars}

First we show a connection between elevator chains and caterpillars.
\begin{lemma}\label{lem:elevatorChainsImplyClosedUndernFoldings}
	Let $n\in\N$ and $\bB$ be a finite structure such that $\Pol(\bB)$ contains operations that form an elevator chain of length $n$. Then $\Csp(\bB)$ is closed under $n$-foldings of caterpillars.
\end{lemma}
\begin{proof}
	Let $\bC$ be a caterpillar with spine $(a_0,\dots,a_m)$ and $0\leq i_1< j_1\leq\dots\leq i_n< j_n\leq m$ such that for each $k\in[n]$ the $(a_{i_k},a_{j_k})$-unfolding of $\bC$ has a homomorphism $h_k$ to $\bB$. Let $e_0,\dots,e_{n+1}\in\Pol(\bB)$ be operations that form an elevator chain of length $n$. Let $e=e_0$ and recall that $e_{n+1}=e$ and that $e_k(c_1,\dots,c_{k-1},c_k,c_k,c_k,c_{k+1},\dots,c_n)=e(c_1,\dots,c_n)$ for all $k\in[n]$ and all $c_1,\dots,c_n\in\bC$.
	Define a map $h\colon\{a_0,\dots,a_m\}\to\bB$ where $h(a_\ell)$ is 
	\begin{itemize}
		\item $e_k(h_1(a_\ell),\dots,h_{k-1}(a_{\ell}),h_k(a'_\ell),h_k(a_\ell),h_k(a''_\ell),h_{k+1}(a_{\ell}),\dots,h_n(a_\ell))$
		if $i_k<\ell<j_k$, where $a_\ell'$ and $a''_\ell$ are the copies of $a_\ell$ in the $(a_{i_k},a_{j_k})$-unfolding of $\bC$ according to the naming convention from the definition of unfoldings and
		\item $e(h_1(a_\ell),\dots,h_n(a_\ell))$ if $\ell$ is in none of the intervals $[i_k+1,j_k-1]$.
	\end{itemize}
	By looking at the definition of $\Ele n$ it is easy to see, that $h$ can be extended to a homomorphism from $\bC$ to $\bB$.
\end{proof}

The remainder of this section is dedicated to proving the implication $(\ref{closedUnderFoldings})\Rightarrow(\ref{caterpillarAndUnfolding})$, by studying the connection between  $(n+1)$-foldings and $n$-fixed unfoldings. 
%------------------------------------------------------------

\begin{example}\label{exa:simpleNumbersGameMotivation}
    Recall that the duality pair $(\{\bP_4\},\T_4)$ witnesses that $\T_4$ has $2$-fixed unfolded caterpillar duality and that $\P_4$ is a caterpillar with spine $(1,2,3)$. Let $\P_{213}$ denote the $(1,2)$-unfolding of $\P_4$. The name $\P_{213}$ comes from the fact that $\P_{231}$ is an oriented path consisting of 2 forward edges, followed by 1 backward edge, followed by 3 forward edges. Using this naming convention it is clear that $\P_{312}$ is the $(2,3)$-unfolding of $\P_4$. 
    Our main theorem predicts that $\overline{\Csp(\T_4)}$ is not closed under 1-fixed unfoldings and that $\Csp(\T_4)$ is not closed under 2-foldings. 
    The former holds, since $\P_{213}\to\T_4$ shows that not all unfoldings of $\P_4$ that fix the vertebra $(2,3)$ are obstructions of $\T_4$ and analogously $\P_{213}\to\T_4$ shows that fixing the vertebra $(1,2)$ is also not sufficient. 
    These structures also witness the latter, since $\P_4$ is a 2-folding of $(\P_{213},\P_{312})$, $\P_{213}\to\T_4$, $\P_{312}\to\T_4$, and $\P_4\not\to\T_4$. 
    
    In this case the connection from 2-foldings to 1-fixed unfoldings is easy. Now consider the paths $\P_{21113}$ and $\P_{31112}$. They still witness that  $\overline{\Csp(\T_4)}$ is not closed under 1-fixed unfoldings. However, for 2-foldings we need to do more work. Note that $\P_{21212}$ is a 2-folding of $(\P_{21113},\P_{31112})$ and that $\P_{21212}\to\T_4$. We can still salvage these witnesses, since $\P_{213}$ is a 2-folding of $(\P_{21113},\P_{21212})$ and $\P_{312}$ is a 2-folding of $(\P_{21212},\P_{31112})$, yielding our original witnesses. 
    It is easy to see that this approach still works when the two initial witnesses are even more unfolded.
\end{example}

In general this can become much more complicated since the number of steps we might have to take increases significantly in $n$ and witnesses might not be as simple as unfolding a single vertebra multiple times. We will deal with the first complication in Lemma~\ref{lem:numbersGameCaterpillarVersion} and with the second one in Lemma~\ref{lem:closedUnderFoldingsImpliesClosedUnderUnfoldings}. Note that in the example we essentially only cared about the number of times a vertebra was unfolded. Hence, we introduce a version of foldings for numbers and translate our findings to caterpillars afterwards.

%\subsection{Foldings of Numbers}
Let $n\in\N$ and let $r_1,\dots,r_n\in \N^n$. The \emph{$n$-folding} of $(r_1,\dots,r_n)$ is the tuple $r\in\N^n$, where the $i$-th entry of $r$ is \[\max((r_1)_i,\dots,(r_{i-1})_i,(r_i)_i-1,(r_{i+1})_i\dots,(r_n)_i).\]% or $0$ if all $(r_j)_i$ are 0.
Let $M\subseteq \N^n$. The set $M$ is \emph{closed under $n$-foldings} if for all $r_1,\dots,r_n\in M$ the $n$-folding of $(r_1,\dots,r_n)$ is contained in $M$.

\begin{lemma}\label{lem:numbersGame}
    Let $n,m\in\N$ and let $M\subseteq \N^n$. If $M$ is closed under $n$-foldings and $(m,0,\dots,0), (0,m,\dots,0),\dots,(0,0,\dots,m)\in M$, then $(a_1,\dots,a_n)\in M$ for all $a_1,\dots,a_n\in[0,m-1]$. In particular $(0,\dots,0)\in M$.
\end{lemma}
It is a nice exercise to show this lemma from hand for the case $n=m=3$.
\begin{proof}
    We order the set $[0,m-1]^n$ lexicographically and show by induction that all its elements are contained in $M$. For the base case consider the following diagram
    \begin{align*}
    \begin{matrix}
        e_1=&(m&0&\dots&0)\\
        e_2=&(0&m&\dots&0)\\
        \vdots&\vdots&\vdots&&\vdots\\
        e_n=&(0&0&\dots&m)\\\hline
        e=&(m-1&m-1&\dots&m-1)
    \end{matrix}
    \end{align*}
    it shows that $e$ is the $n$-folding of $(e_1,\dots,e_n)$. Hence $e\in M$.
    For the induction step let $i\in[n]$ and $a_i,\dots,a_n\in[1,m-1]$.
    Consider the diagram
    \begin{align*}
    \begin{matrix}
        r_1=&(a_1+1&\cdots&a_{i-1}&a_i&0&\cdots&0)\\
        \vdots&\vdots&&\vdots&\vdots&\vdots&&\vdots\\
        r_{i-1}=&(a_1&\cdots&a_{i-1}+1&a_i&0&\cdots&0)\\
        r_i=&(a_1&\cdots&a_{i-1}&a_i&0&\cdots&0)\\
        r_{i+1}=&(0&\cdots&0&0&m&\cdots&0)\\
        \vdots&\vdots&&\vdots&\vdots&\vdots&&\vdots\\
        r_n=&(0&\cdots&0&0&0&\cdots&m)\\\hline
        r=&(a_1&\cdots&a_{i-1}&a_i-1&m-1&\cdots&m-1)\\
    \end{matrix}
    \end{align*}
    it shows that $r$ is the $n$-folding of $(r_1,\dots,r_n)$. 
    Note that $r_i$ is the (lexicographic) successor of $r$. 
    We want to show $r\in M$. 
    For $j\in[n]$ define $r'_j$ to be 
    \begin{itemize}
        \item $r_j$ if $j\leq i$ and $a_j+1< m$ and
        \item $e_j$ otherwise. %(0,\dots,0,m,0,\dots,0)
    \end{itemize}
    Note that $r$ is the $n$-folding of $(r'_1,\dots,r'_n)$. Furthermore for every $j\in[n]$ we have that either
    \begin{itemize}
        \item $r'_j$ is $e_j$, which is in $M$ or
        \item $r'_j$ (lexicographically) indirectly succeeds $r$ and  $r'_j\in [m-1]^n$, hence $r'_j\in M$ by induction hypothesis.
    \end{itemize}  
    Since $M$ is closed under $n$-foldings we have that $r\in M$.
\end{proof}

Now we translate the previous result about $n$-foldings of tuples to $n$-foldings of caterpillars. 
\begin{lemma}\label{lem:numbersGameCaterpillarVersion}
    Let $\bB$ be a finite structure whose CSP is closed under $n$-foldings of caterpillars, $\bC$ be a caterpillar with spine $(a_0,\dots,a_m)$ and vertebrae $v_1=(a_0,a_1),\dots,v_m=(a_{m-1},a_m)$, and $0\leq i_1\leq j_1\leq\dots\leq i_n\leq j_n\leq m$. 
    Let $\bC_1,\dots,\bC_n$ be caterpillars such that for each $\ell\in[n]$ we have that $\bC_\ell$ is an unfolding of $\bC$ that fixes all vertebrae of $\bC$ that are not in $\bC[i_\ell,j_\ell]$. 
%    \begin{itemize}
%        \item $\bC_1$ is an unfolding of $\bC$ that fixes $v_{i_1+1},\dots,v_m$,
%        \item $\bC_2$ is an unfolding of $\bC$ that fixes $v_1,\dots,v_{i_{1}},v_{i_2+1},\dots,v_m$,
%        \item $\dots$, 
%        \item $\bC_{n-1}$ is an unfolding of $\bC$ that fixes $v_1,\dots,v_{i_{n-2}},v_{i_{n-1}+1},\dots,v_m$, and
%        \item $\bC_n$ is an unfolding of $\bC$ that fixes $v_1,\dots,v_{i_{n-1}}$,
%    \end{itemize}
If $\bC_1,\dots,\bC_n\to\bB$, then $\bC\to\bB$.
\end{lemma}
\begin{proof}
If there is an $\ell\in[n]$ with $i_\ell=j_\ell$, then $\bC=\bC_\ell\to\bB$. Hence, we can assume that $i_\ell<j_\ell$ for all $\ell$. 
For each $\ell\in[n]$ let $(a^\ell_1,b^\ell_1),\dots,(a^\ell_{k(\ell)},b^\ell_{k(\ell)})$ be such that the sequence \[\bC=\bC_{(0,\dots,0,\dots,0)},\bC_{(0,\dots,1,\dots,0)},\dots,\bC_{(0,\dots,k(\ell),\dots,0)}=\bC_\ell\] is well defined, where $\bC_{(0,\dots,i+1,\dots,0)}$ is the $(a^\ell_{i+1},b^\ell_{i+1})$-unfolding of $\bC_{(0,\dots,i,\dots,0)}$.
Note that since $i_\ell<j_\ell$ for all $\ell$ we could replace each $\bC_\ell$ by a further unfolded version of itself. Hence we can assume without loss of generality that all $k(\ell)$ are equal to $k\coloneqq k(1)$.
We can extend the definition of $\bC_{(0,\dots,i,\dots,0)}$ to tuples with more than one non-zero entry:
For $(k_1,\dots,k_n)\in[0,k]^n$ and $\ell\in[n]$ define $\bC_{(k_1,\dots,k_\ell+1,\dots,k_n)}$ as the $(a^\ell_{k_\ell+1},b^\ell_{k_\ell+1})$-unfolding of $\bC_{(k_1,\dots,k_\ell,\dots,k_n)}$. 

Let $M$ be the set of all tuples $t\in[0,k]^n$ such that $\bC_t\to\bB$. By definition \[(k,0,\dots,0),(0,k,\dots,0),\dots,(0,0,\dots,k)\in M.\] We will show that $M$ is closed under $n$-foldings.
Let $r_1,\dots,r_n\in M$ and let $r$ be the $n$-folding of $(r_1,\dots,r_n)$. For $\ell\in[n]$ let $r'_\ell$ be the tuple obtained from $r$ by increasing the $\ell$-th entry by 1. Note that $r_\ell\leq r'_\ell$ componentwise. Hence, $\bC_{r'_\ell}\to\bC_{r_\ell}\to\bB$. Observe that $\bC_r$ is the $n$-folding of $(\bC_{r'_1},\dots,\bC_{r'_n})$. Therefore, $\bC_r\to\bB$ and $r\in M$. By Lemma~\ref{lem:numbersGame} $(0,\dots,0)\in M$. Hence, $\bC=\bC_{(0,\dots,0)}\to\bB$, as desired.
\end{proof}

%------------------------------------------------------------

%\subsection{Common Unfoldings of Caterpillars}
Next we introduce some convenient notation for manipulating caterpillars. 
Let $\bC$ be a caterpillar with conjunctive query representation $(\phi\wedge\psi)(a_0,\dots,a_m)$ and let $0\leq i\leq j\leq m$. The \emph{$[i,j]$-interval of $\bC$}, denoted $\bC[i,j]$, is the canonical database of \[(\phi_i(a_i,a_{i+1})\wedge\dots\wedge\phi_{j-1}(a_{j-1},a_j))\wedge(\psi_i(a_i)\wedge\dots\wedge\psi_j(a_j)).\] Note that $\bC[i,j]$ is a caterpillar. The spine of $\bC[i,j]$ is $(a_i,\dots,a_j)$. Recall that $(a_j,\dots,a_i)$ is also the spine of $\bC[i,j]$. We denote $\bC[i,j]$ with this \emph{reversed spine} by $\bC[j,i]$. We define the \emph{$(i,j)$-interval of $\bC$}, denoted $\bC(i,j)$, as the canonical database of \[(\phi_i(a_i,a_{i+1})\wedge\dots\wedge\phi_{j-1}(a_{j-1},a_j))\wedge(\psi_{i+1}(a_{i+1})\wedge\dots\wedge\psi_{j-1}(a_{j-1})).\] Analogously we define $\bC(i,j]$ and $\bC[i,j)$. 
Using this notation we can show that any unfoldings of a caterpillar have a common unfolding, similar to a common multiple for numbers.

\begin{lemma}\label{lem:commonUnfoldingExists}
	Let $\bC$ be a caterpillar and let $\bC_1$ and $\bC_2$ be unfoldings of $\bC$ fixing $v_1,\dots,v_n$. Then there exists a caterpillar $\bC'$ that is an unfolding of $\bC_1$ and of $\bC_2$ fixing $v_1,\dots,v_n$.
\end{lemma}
\begin{proof}
	Let $\bC$ be a caterpillar with spine $(a_0,\dots,a_m)$ and vertebrae $v_1=(a_0,a_1),\dots,v_m=(a_{m-1},a_m)$, let $0\leq i(1)<\dots<i(n)\leq m$, and let $\bC_1$ and $\bC_2$ be unfoldings of $\bC$ fixing $v_{i(1)},\dots,v_{i(n)}$. Let $h_1\colon \bC_1\to\bC$ and $h_2\colon \bC_2\to\bC$ be the homomorphisms from the definition of unfoldings and let $(a^1_0,\dots,a^1_{m(1)})$ and $(a^2_0,\dots,a^2_{m(2)})$ be the spines of $\bC_1$ and $\bC_2$, respectively. 
	We proceed with an induction on $m$. If $m=0$, then $\bC_1=\bC_2=\bC$ and $\bC$ is a common unfolding of both. For $m>0$ there are two cases to consider, see Figure~\ref{fig:commonUnfoldingCases} for a visualization. 
	\begin{figure}
		\centering
		\begin{tikzpicture}
			
			\node at (0,0) {\begin{tikzpicture}[scale=0.9,
					bullet/.style={circle, fill, inner sep=1.5pt}
					]
					
					% --- Left path ---
					\coordinate (L0) at (0.5,0);
					\coordinate (L1) at (0.5,0.7); 
					\coordinate (L2) at (2.5,3);
					\node at (1.5,-1) {$\bC_1$};
					
					\path (L0) edge (L1);
					\draw plot [smooth] coordinates {(L1) (0.75,1.6) (1.2,0.7) (1.7,2.4) (2.2,1.7)  (L2)};
					
					\node[bullet, label=left:$a^1_0$] at (L0) {};
					\node[bullet, label=left:$a^1_{1}$] at (L1) {};
					\node[bullet, label=right:$a^1_{m(1)}$] at (L2) {};
					
					% --- Right path ---
					\coordinate (R0) at (3,0);
					\coordinate (R1) at (3.0,0.7);   % same height as L1
					\coordinate (R2) at (5.0,3);
					\node at (4,-1) {$\bC_2$};
					
					\path (R0) edge (R1);
					\draw plot [smooth] coordinates {(R1) (3.3,2.5) (3.7,1.2) (4.2,2.8) (4.6,2.3)  (R2)};
					
					\node[bullet, label=left:$a^2_0$] at (R0) {};
					\node[bullet, label=left:$a^2_{1}$] at (R1) {};
					\node[bullet, label=right:$a^2_{m(2)}$] at (R2) {}; % unlabeled top endpoint
					
			\end{tikzpicture}};

			\node at (7.5,0) {
				\begin{tikzpicture}[scale=0.9,
					bullet/.style={circle, fill, inner sep=1.5pt}
					]
					
					% --- Left path ---
					\coordinate (L0) at (0.5,0);
					\coordinate (L1) at (1.5,1.5);   % a_k (aligned height)
					\coordinate (L2) at (2.5,0);
					\coordinate (L3) at (4.5,3);
					\node at (2.5,-1) {$\bC_1$};
					
					\draw plot [smooth] coordinates {(L0) (0.75,1.3) (1.2,0.5)  (L1)};
					\draw plot [smooth] coordinates {(L1) (1.8,0.7) (2.2,1.1)  (L2)};
					\draw plot [smooth] coordinates {(L2) (2.9,1.4) (3.2,1.0) (3.5,2.4) (4.2,1.7)  (L3)};
					
					\node[bullet, label=left:$a^1_0$] at (L0) {};
					\node[bullet, label=above:$a^1_{\ell(1)}$] at (L1) {};
					\node[bullet, label=left:$a^1_k$] at (L2) {};
					\node[bullet, label=right:$a^1_{m(1)}$] at (L3) {};
					
					% --- Right path ---
					\coordinate (R0) at (5,0);
					\coordinate (R1) at (6,1.5);   % same height as L1
					\coordinate (R2) at (7.5,3);
					\node at (6.25,-1) {$\bC_2$};
					
					\draw plot [smooth] coordinates {(R0) (5.25,1.2) (5.8,0.6)  (R1)};
					\draw plot [smooth] coordinates {(R1) (6.3,2.7) (6.7,1.0) (7.0,2.8) (7.25,2.3)  (R2)};
					
					\node[bullet, label=left:$a^2_0$] at (R0) {};
					\node[bullet, label=left:$a^2_{\ell(2)}$] at (R1) {};
					\node[bullet, label=right:$a^2_{m(2)}$] at (R2) {}; % unlabeled top endpoint
					
			\end{tikzpicture}};
			
		\end{tikzpicture}
		\caption{The two cases in the proof of Lemma~\ref{lem:commonUnfoldingExists}. The first case is on the left side and second is on the right.}
		\label{fig:commonUnfoldingCases}
	\end{figure}
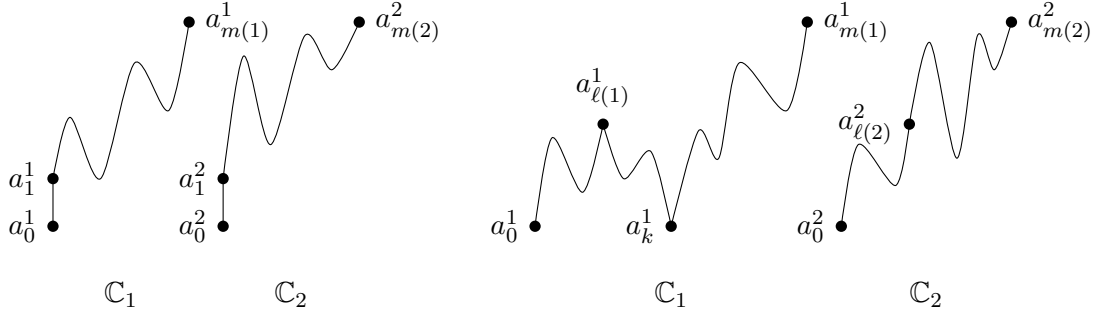
	In the first case $h_1^{-1}(a_0)=\{a^1_0\}$ and $h_2^{-1}(a_0)=\{a^2_0\}$. Then $h_i(a^i_k)\in\{a_1,\dots,a_m\}$ for all $i\in\{1,2\}$ and all $k\in[m_i]$. Therefore $\bC_1[1,m(1)]$ and $\bC_2[1,m(2)]$ are unfoldings of $\bC[1,m]$ fixing $v_{i(1)},\dots,v_{i(n)}$. By induction there is a common unfolding $\bC'$ of $\bC_1[1,m(1)]$ and $\bC_2[1,m(2)]$ that fixes $v_{i(1)},\dots,v_{i(n)}$. Replacing the interval $\bC[1,m]$ in $\bC$ by $\bC'$ yields the desired common unfolding of $\bC_1$ and $\bC_2$. 
	The second case is that either $|h_1^{-1}(a_0)|\geq2$ or $|h_2^{-1}(a_0)|\geq2$. Assume without loss of generality that $h_1^{-1}(a_0)=\{a^1_0,a^1_k,\dots,\}$. Let $\ell$ be the largest number in $[m]$ such that there exists an $\ell(1)\in [0,k]$ with $h_1(a^1_{\ell(1)})=a_\ell$. Let $\ell(2)$ be the smallest number in $[0,m(2)]$ with $h_2(a^2_{\ell(2)})=a_\ell$. Note that by definition  $v_{i(1)},\dots,v_{i(n)}$ must be in the interval $\bC[\ell,m]$ and that
	\begin{itemize} 
		\item $\bC_1[0,\ell(1)]$ and $\bC_2[0,\ell(2)]$ are unfoldings of $\bC[0,\ell]$, 
		\item $\bC_1[\ell(1),k]$ and $\bC_2[\ell(2),0]$ are unfoldings of $\bC[\ell,0]$, and
		\item $\bC_1[\ell(1),m(1)]$ and $\bC_2[0,m(2)]$ are unfoldings of $\bC[0,m]$ fixing $v_{i(1)},\dots,v_{i(n)}$.
	\end{itemize}
	By induction there are common unfoldings $\bC'_1,\bC'_2,\bC'_3$. By taking $\bC_1$ and replacing 
	\begin{align*}
		&\text{$\bC_1[0,\ell(1)]$ by $\bC'_1$,}&&
		\text{$\bC_1[\ell(1),k]$ by $\bC'_2$,} &&\text{and}&&
		\text{$\bC_1[k,m(1)]$ by $\bC'_3$}
	\end{align*}
	we obtain a caterpillar $\bC'$. By construction $\bC'$ is the desired  common unfolding of $\bC_1$ and $\bC_2$ fixing $v_{i(1)},\dots,v_{i(n)}$.
	Finally, observe that if $v_{i(1)}=v_1$, then $h_1^{-1}(a_0)=\{a^1_0\}$ and $h_2^{-1}(a_0)=\{a^2_0\}$. Hence, we are in the first case and the constructed common unfolding fixes $v_1$.
\end{proof}

%\subsection{From Foldings to Unfoldings}

\begin{lemma}\label{lem:closedUnderFoldingsImpliesClosedUnderUnfoldings}
    Let $n\in\N$ and $\bB$ be a finite structure whose CSP is closed under $(n+1)$-foldings of caterpillars. Then $\overline{\Csp(\bB)}$ is closed under $n$-fixed unfoldings of caterpillars.
\end{lemma}
\begin{proof}
    Let $\bC$ be a caterpillar with spine $(a_0,\dots,a_m)$ and vertebra $v_1=(a_0,a_1),\dots,v_m=(a_{m-1},a_m)$ such that for all $1\leq i_1<\dots<i_n\leq m$ there is an unfolding $\bC_{\{v_1,\dots,v_n\}}$ of $\bC$  fixing $v_{i_1},\dots,v_{i_n}$ such that $\bC_{\{v_1,\dots,v_n\}}\to\bB$. As discussed in Example~\ref{exa:simpleNumbersGameMotivation}, it suffices to show $\bC\to\bB$ to prove the lemma. Let $V\coloneqq\{v_1,\dots,v_m\}$ and $\mathcal I$ be the set of all subsets $I$ of $V$ such that there is an unfolding $\bC_I$ of $\bC$ that fixes all vertebrae in $I$ and $\bC_I\to\bB$. Note that $\mathcal I$ is closed under taking subsets.
    %We order the subsets of $V$ in the following way, for two distinct sets $I,I'\subeteq V$ we have $I<I'$ if and only if the $v_i$ with the smallest $i$ in the symmetric difference of $I$ and $I'$ is in $I$.
    For each $I\subseteq V$ the \emph{characteristic tuple} $t_I\in\{0,1\}^m$ is the tuple that has for all $i\in[m]$ a 1 in the $i$-th entry if and only if $v_i\in I$. We order the subsets of $V$ by comparing their characteristic tuples lexicographically. 
    We will show by induction that $\mathcal I$ contains all subsets of $V$.
    
    Let $n'\in[0,m]$ and $1\leq i(1)<\dots<i(n')\leq m$. Define $I\coloneqq\{v_{i(1)},\dots,v_{i(k)}\}$. If $n'\leq n$, then $I\in\mathcal I$ by choice of $\bC$. 
    If $n'\geq n+1$, then define \[I_1\coloneqq {I\setminus\{v_{i(1)}\}}\cup\{v_{i(1)+1},\dots,v_m\},\dots, I_{n+1}\coloneqq {I\setminus\{v_{i({n+1})}\}}\cup\{v_{i({n+1})+1},\dots,v_m\}.\] 
    By definition $t_{I_1},\dots,t_{I_{n+1}} <_{\text{lex}} t_I$. Hence, by induction assumption, there are $\bC_{I_1},\dots,\bC_{I_{n+1}}$ witnessing that $I_1,\dots, I_{n+1}\in\mathcal I$, respectively. Let \[(a^1_0,\dots,a^1_{m(1)}),\dots,(a^{n+1}_0,\dots,a^{n+1}_{m({n+1})})\] be the spines of $\bC_{I_1},\dots,\bC_{I_{n+1}}$, respectively. 
    For $k\in [2,n+1]$ let $i(k,1)$ be the smallest number in $[m(k)]$ such that $a^k_{i(k,1)}=a_{i(1)}$. Observe that $\bC_{I_k}[0,i(k,1)]$ is an unfolding of $\bC[0,i(1)]$ that fixes $v_{i(1)}$. By Lemma~\ref{lem:commonUnfoldingExists} there is a caterpillar that is a common unfolding of $\bC_{I_2}[0,i(2,1)],\dots,\bC_{I_{n+1}}[0,i({n+1},1)]$ fixing ${v_{i(1)}}$. We denote this caterpillar by $\bC'[0,i'(1)]$, where $i'(1)$ is the number of vertebrae of $\bC'[0,i'(1)]$. Note that $\bC'[0,i'(1)]$ is in particular an unfolding of $\bC[0,i(1)]$ that fixes $v_{i(1)}$.
    Let $i(1,1)$ be the biggest number in $[m(1)]$ such that $a^1_{i(1,1)}=a_{i(1)}$. Since $\bC_{I_1}$ is an unfolding that fixes $v_{i(1)+1}$ we have that $a_{i(1)+1}$ does not occur in $\bC_{I_1}[0,i(1,1)]$. Hence, $\bC_{I_1}[0,i(1,1)]$ is an unfolding of $\bC[0,i(1)]$ (that does not necessarily fix $v_{i(1)}$). By Lemma~\ref{lem:commonUnfoldingExists} there is a caterpillar $\bC''[0,i''(1)]$ that is a common unfolding of $\bC_{I_1}[0,i(1,1)]$ and $\bC'[0,i'(1)]$, where $i''(1)$ is the number of vertebrae of $\bC''[0,i''(1)]$.
    We continue in this fashion to define caterpillars $\bC'[i'(1),i'(2)],\dots,\bC'[i'(n+1),m]$ and $\bC''[i''(1),i''(2)],\dots,\bC''[i''(n+1),m]$, where analogous to the case $\ell=1$ we have that $i(k,\ell)$ is the smallest number in $[m(k)]$ such that $a^k_{i(k,\ell)}=a_{i(\ell)}$ if $\ell\neq k$ and it is the biggest such number if $k=\ell$. Observe that for $k\neq\ell$ the fact that $\bC_{I_k}$ fixes $v_{i(\ell-1)}$ and $v_{i(\ell)}$ implies that $\bC'[i'(k,\ell-1),i'(k,\ell)]$ is an unfolding of $\bC_{I_k}[i(\ell-1),i(\ell)]$ that fixes $v_{i(\ell)}$. 
    Similarly, the fact that $\bC_{I_k}$ fixes $v_{i(k-1)}$ and $v_{i(k)+1}$ implies that $\bC''[i''(k,k-1),i''(k,k)]$ is an unfolding of $\bC_{I_k}[i(k-1),i(k)]$. 

    Note that $\bC'[0,i'(1)],\dots,\bC'[i'(n+1),m]$ uniquely defines a caterpillar $\bC'$. Analogously, $\bC''[0,i''(1)],\dots,\bC''[i''(n+1),m]$ defines $\bC''$. 
    The caterpillar $\bC'$ is an unfolding of $\bC$ that fixes all vertebrae in $I$. 
    For $k\in[n+1]$ the caterpillar $\bC'_{k}$ is obtained from $\bC'$ by replacing the interval $\bC'[i'(k,k-1),i'(k,k)]$ by $\bC''[i''(k,k-1),i''(k,k)]$. Observe that $\bC'_{k}$ is an unfolding of $\bC_{I_k}$ (that fixes $I\setminus\{v_k\}$). Hence $\bC'_k\to\bC_{I_k}\to\bB$. By definition $\bC'_k$ is an unfolding of $\bC'$ that fixes all vertebra of $\bC'$ that are not in the interval $\bC'[i'(k,k-1),i'(k,k)]$. Hence, Lemma~\ref{lem:numbersGameCaterpillarVersion} implies $\bC'\to\bB$.
    Therefore, $I\in\mathcal I$.

    This induction proof shows in particular that there is an unfolding $\bC_V$ of $\bC$ fixing all vertebrae in $V$ witnessing that $V\in\mathcal I$. By definition of $\mathcal I$ there is a homomorphism from $\bC_V$ to $\bB$. Note that the only unfolding of $\bC$ that fixes all vertebrae of $\bC$ is $\bC$ itself. Hence, $\bC_V=\bC$ and $\bC\to\bB$, as desired.
    %Note that $\bC_{\{v_1,\dots,v_m\}}$ is unique and equal to $\bC$. 
\end{proof}

%------------------------------------------------------------

\section{Monadic Datalog}
Let $\bB$ be a structure and $\Pi$ be a monadic Datalog program that solves $\Csp(\bB)$. It seems intuitively clear that for every structure $\bA$ with $\bA\not\to\bB$ there is a derivation of $\Pi$ that derives the goal predicate and is in some sense connected. In this section we will formalize this intuition.
A rule $\phi_0 \; {:}{-} \; \phi$ of a Datalog program is called \emph{injective} if the $\tau$-reduct of the canonical database of $\phi$ is injective, where $\tau$ is the set of EDBs. 
A rule $\phi_0 \; {:}{-} \; \phi$ is called \emph{connected} if \begin{itemize}
	\item the canonical databases of $\phi$ is connected,
	\item every variable in $\phi_0$ occurs in $\phi$, 
	\item there is no 0-ary IDB occurring in $\phi$, and
	\item if $\phi_0$ is 0-ary, then $\phi_0$ is the goal predicate. % or has arity at least one.
\end{itemize} %Note that whenever a Datalog program derives the goal predicate on a structure $\bA$ with a derivation that uses only connected rules and all rules were neccessary, then all IDBs have been added to a connected component (whats with 0 ary IDBs?)
Let $\Pi$ be a Datalog program and $\bA$ be a finite structure. A derivation $\bA_0\vdash_{R_0} \dots\vdash_{R_{n-1}} \bA_n$ of $\Pi$ on $\bA$ that derives the IDB $P$ on the tuple $t$ is \emph{inefficient} if there is an $m<n$, $0\leq i_1<\dots<i_{m-1}<n$ and a derivation $\bA'_0\vdash_{R_{i_0}} \dots\vdash_{R_{i_{m-1}}} \bA'_m$ of $\Pi$ on $\bA$ that also derives $P$ on $t$. 
A derivation is called \emph{efficient} if it is not inefficient. Note that for every derivation of $\Pi$ on $\bA$ deriving $P$ on $t$ there is an efficient derivation of $\Pi$ on $\bA$ deriving $P$ on $t$.
It is clear from these definitions that the following is true.
\begin{lemma}
	%conncted rules + no 0 ary IDBs implies goal derivation on connected components
	Let $\Pi$ be a Datalog program and $\bA$ be a finite structure. If $\bA_0\vdash_{R_0} \dots\vdash_{R_{n-1}} \bA_n$ is an efficient derivation of $\Pi$ on $\bA$ with connected rules $R_0,\dots,R_{n-1}$,
	then there is a set $S\subseteq A$ such that 
	\begin{itemize}
		\item $\bA[S]$ is connected and
		\item $\bA_0[S]\vdash_{R_0} \dots\vdash_{R_{n-1}} \bA_n[S]$ is a derivation of $\Pi$ on $\bA[S]$,
	\end{itemize}
	where $\bA[S]$ is the substructure of $\bA$ induced by $S$.
\end{lemma}

We show that a monadic Datalog program that solves a CSP of a finite structure can basically always restrict to injective and connected rules, when deriving the goal predicate.

\begin{lemma}\label{lem:DLinjective}
	Let $\bB$ be a finite structure and $\Pi$ be a Datalog program that solves $\Csp(\bB)$. Let $\Pi'$ be the Datalog program obtained from $\Pi$ by removing all  non-injective rules. Then $\Pi'$ solves $\Csp(\bB)$.
\end{lemma}
\begin{proof}
	Let $\bA$ be an instance of $\Csp(\bB)$. First, suppose $\bA\to\bB$. Since $\Pi$ solves $\Csp(\bB)$, it can not derive the goal predicate on $\bA$. As $\Pi'$ has fewer rules than $\Pi$, it also can not derive the goal predicate on $\bA$.
	Now consider the case $\bA\not\to\bB$. Clearly, the set of all finite structures $\bA'$ with $\bA'\not\to\bB$ is an obstruction set for $\bB$. Hence, by Lemma~\ref{lem:inj}, there exists an injective structure $\bA'$ with $\bA'\to\bA$ and $\bA'\not\to\bB$. Since $\Pi$ solves $\Csp(\bB)$, it can derive the goal predicate on $\bA'$. Note that, since $\bA'$ is injective, all rules of this derivation are injective. Hence, this derivation is also a derivation of $\Pi'$ that derives the goal predicate. Applying the homomorphism from $\bA'$ to $\bA$ to this derivation yields a derivation of $\Pi'$  on $\bA$ that derives the goal predicate.
\end{proof}

\begin{lemma}\label{lem:DLconnected}
	Let $\bB$ be a (not necessarily finite) structure and $\Pi$ be a monadic Datalog program that solves $\Csp(\bB)$. 
	Then there is a monadic Datalog program $\Pi'$ such that
	\begin{itemize}
		\item $\Pi'$ solves $\Csp(\bB)$, 
		%\item for every finite structure $\bA$ with $\bA\not\to\bB$ there is a derivation of $\Pi'$ on a connected induced substructure $\bA'$ of $\bA$ that derives the goal predicate and all rules applied in this derivation are connected and injective, and
		\item $\Pi'$ is obtained from $\Pi$ %by removing all non-injective rules and 
		replacing each disconnected rule $\phi_0 \; {:}{-} \;\phi_1\wedge\dots\wedge\phi_n$ with a suitable set of connected rules of the forms \begin{align*}&\phi_0 \; {:}{-} \;\phi_{i_1}\wedge\dots\wedge\phi_{i_m}&&\text{ and }&&G \; {:}{-} \;\phi_{i_1}\wedge\dots\wedge\phi_{i_m}.\end{align*}
	\end{itemize}
	If $\Pi$ is symmetric, then the suitable sets can be chosen in such a way that $\Pi'$ is also symmetric.
\end{lemma}
%Note that, by definition of connected rules, the only IDB of $\Pi'$ with arity 0 is the goal predicate.
\begin{proof}
	We iteratively construct Datalog programs $\Pi_i$ that solves $\Csp(\bB)$. Let $\Pi_0=\Pi$. %Let $\Pi_0$ be the Datalog program obtained from $\Pi$ by removing all rules that are not injective. By Lemma~\ref{lem:DLinjective} $\Pi_0$ still solves $\Csp(\bB)$. 
	For $i\geq0$ we construct $\Pi_{i+1}$ as follows. Define $\Pi_{i+1}\coloneqq \Pi_i$ if there is no finite structure $\bA$ with $\bA\not\to\bB$ such that every derivation of $\Pi_i$ on $\bA$ that derives the goal predicate uses a disconnected rule. 
	Otherwise, let $\bA$ be such a structure and let $\bA_0\vdash_{R_0}\dots\vdash_{R_{n}}\bA_{n+1}$ be a derivation of $\Pi_i$ that 
	\begin{itemize}
		\item derives the goal predicate on $\bA$,
		\item uses a disconnected rule $R_k$, and
		\item among all derivations of $\Pi_i$ on $\bA$ that derive the goal predicate, minimizes the number of applications of disconnected rules.%the number of applications of disconnected rules in this derivation is minimal among all derivations of $\Pi_i$ on $\bA$ that derive the goal predicate.
	\end{itemize} 
	Let $R_k$ be of the form $\phi_0 \; {:}{-} \;\phi$. 
	%Observe that if $\phi_0$ is a 0-ary IDB $P$ that is not the goal predicate, then, by minimality of the derivation, there must be a rule $R_{k'}$ with $k'>k$ that has $P$ in its body. By definition $R_{k'}$ is disconnected.  
	%Hence, we can without loss of generality choose a disconnected rule $R_k$ where $\phi_0$ is either $G$ or its IDB is of arity at least 1.
	Let $\mathbb R$ denote  the canonical database of $\phi$. %, otherwise, let $\R$ denote the canonical database of $\phi_0\wedge\phi$. 
	Let $\psi$ be the empty conjunction if $\phi_0=G$, otherwise $\phi_0$ is of the form $P(x)$ and we define $\psi$ to be the canonical conjunctive query of the connected component of $\mathbb R$ that contains $x$. 
	%Since $\R_k$ is disconnected, $\R$ has at least two connected components. 
	Let $\psi_1,\dots,\psi_m$ be the canonical conjunctive queries of the remaining connected components of $\mathbb R$.
	%that do not contain the canonical database of $\phi_0$ and let $\psi$ be the remaining component (which is empty when $\phi_0=G$). 
	The formulas $\psi,\psi_1,\dots,\psi_m$ partition the conjuncts of $\phi$. 
	We distinguish two cases. 
	
	First, assume that for every $j\in[m]$, there exists a finite structure $\bD_j$ such that $\bD_j\to\bB$ and there is a derivation of $\Pi_i$ on $\bD_j$ that results in a structure $\bC_j$ to which the canonical database of $\psi_j$ admits a homomorphism. 
	Then $\Pi_{i+1}$ is obtained from $\Pi_{i}$ by adding the connected rule $\phi_0 \; {:}{-} \;\psi$, which we denote by $\tilde R_k$. This rule was not already a rule of $\Pi_i$, since otherwise \[\bA_0\vdash_{R_0}\dots\vdash_{R_{k-1}}\bA_k\vdash_{\tilde R_k}\bA_{k+1}\vdash_{R_{k+1}}\dots \vdash_{R_{n}}\bA_{n+1}\]
	would be a derivation of $\Pi_i$ that derives the goal predicate on $\bA$ and uses one fewer disconnected rules than the original derivation, contradicting its minimality. 
	Assume there exists a structure $\bA'$ on which $\Pi_{i+1}$ can derive the goal predicate and $\Pi_i$ can not.   Since $\Pi_i$ solves $\Csp(\bB)$, we have $\bA'\to\bB$. It is easy to see that $\Pi_i$ can derive the goal predicate on the disjoint union of $\bD_1,\dots,\bD_m$, and $\bA'$ by replacing $\tilde R_k$ with $R_k$ in the derivation of $\Pi_{i+1}$ and, if necessary, adding parts of the original derivation. Since each of these structures admits a homomorphism to $\bB$ so does their disjoint union. This contradicts the fact that $\Pi_i$ solves $\Csp(\bB)$. Hence, such an $\bA'$ can not exist, and therefore $\Pi_{i+1}$ solves $\Csp(\bB)$. Note that if $\phi_0=G$ we can not be in this case.
	
	In the second case there is a $j$ for which no such structure $\bD_j$ exists.
	Then $\Pi_{i+1}$ is obtained from $\Pi_i$ by adding the connected rule $G \; {:}{-} \;\psi_j$, which we denote by $\tilde R_k$. Again, this rule was not already a rule of $\Pi_i$, since otherwise \[\bA_0\vdash_{R_0}\dots\vdash_{R_{k-1}}\bA_k\vdash_{\tilde R_k}\tilde\bA_{k+1}\]
	where $\tilde\bA_{k+1}$ is obtained from $\bA_{k}$ by adding $()$ to $G^{\bA_k}$, would be a derivation of $\Pi_i$ that derives the goal predicate on $\bA$ that uses at least one fewer disconnected rules than the original derivation, again contradicting its minimality. 
	Assume that there is a structure $\bA'$ on which $\Pi_{i+1}$ can derive the goal predicate and $\Pi_i$ can not. Since $\Pi_i$ solves $\Csp(\bB)$ we have $\bA'\to\bB$. Clearly, $\Pi_{i+1}$ must use $\tilde R_k$ in its derivation of the goal predicate. Therefore,  there is a derivation of $\Pi_i$ on $\bA'$ such that the canonical database of $\psi_j$ has a homomorphism into the derived structure. A contradiction.  
	%It is easy to see that $\Pi_i$ can, by replacing $\tilde R_k$ with $R_k$ and possibly adding parts of the original derivation, derive the goal predicate on the disjoint union of $\bD$ and $\bA'$. Since both structures have a homomorphism to $\bB$ so does their disjoint union. This contradicts that $\Pi_i$ solves $\Csp(\bB)$. 
	Therefore, such an $\bA'$ can not exist, and hence $\Pi_{i+1}$ solves $\Csp(\bB)$.

	Note that all rules we add in this process are connected. Therefore, the disconnected rules do not change and there are only finitely many rules that we could add. Hence, the process eventually stabilizes, and there exists an $i$ for which the Datalog program $\Pi_i$ never has to use any of its disconnected rules to derive the goal predicate. We obtain the desired program $\Pi'$ by removing all disconnected rules from $\Pi_i$. 
	
	If $\Pi$ is symmetric we can modify this procedure. Whenever we add a rule in the first case with an IDB in its body we also add the reversed rule. Since a symmetric Datalog programs does not need to contain the reverse of rules of the form $G \; {:}{-} \;\psi_j$, we have that $\Pi_i$ is symmetric. %Only in the first case 
	The argument that $\Pi_i$ solves $\Csp(\bB)$ is analogous to the original argument. We replace $R_k$ and $\tilde R_k$ by their reverse rule and $\bD_1,\dots,\bD_m$ stay the same. %In the second case we do not have to do anything, since a symmetric linear Datalog does not need to contain the reverse of rules of the form $G \; {:}{-} \;\psi_j$. 
\end{proof}

\section{Symmetrizing Linear Arc Monadic Datalog}
%From Caterpillars to Datalog and back again
%join with next section (maybe use subsections?)
The following section is devoted to the implication  from (\ref{caterpillarAndUnfolding}) to (\ref{can-sym-lin-arc}) in the proof of Theorem~\ref{thm:main}.
%------------------------------------------------------------
Before we come to Datalog we introduce an alternative way of defining unfoldings.
Let $\bC_1$ and $\bC_2$ be caterpillars with spines $(a^1_0,\dots,a^1_{m(1)})$ and $(a^2_0,\dots,a^2_{m(2)})$, respectively. The \emph{concatenation of $\bC_1$ and $\bC_2$}, denoted $\bC_1\circ\bC_2$, is the caterpillar obtained from the disjoint union of $\bC_1$ and $\bC_2$ by identifying $a^1_{m(1)}$ and $a^2_0$.  
Note that for a caterpillar $\bC$ with spine $(a_0,\dots,a_n)$ and $i,j\in[0,n]$ with $i<j$ the $(a_i,a_j)$-unfolding of $\bC$ is isomorphic to $\bC[0,j)\circ\bC(j,i)\circ\bC(i,n]$.
 
Let $\bC$ be a caterpillar with conjunctive query representation $(\phi\wedge\psi)(a_0,\dots,a_n)$ and let $n_1,\dots,n_k\in[0,n]$ with $k$ even and $0<n_1>n_2<n_3>n_4<\cdots >n_k<n$. 
The \emph{$(n_1,\dots,n_k)$-unfolding of $\bC$} is $\bC[0,n_1)\circ\bC(n_1,n_2)\circ\bC(n_2,n_3)\circ\bC(n_3,n_4)\circ\dots\circ\bC(n_k,n]$. This unfolding fixes $(a_i,a_{i+1})$ if there is an $\ell\in[0,k]$ such that $n_1,\dots,n_\ell\leq i$ and $n_{\ell+1},\dots,n_k\geq i+1$.
Note that for $i,j\in[0,n]$ with $i<j$ the $(j,i)$-unfolding of $\bC$ is isomorphic to the $(a_i,a_j)$-unfolding of $\bC$ and that for both notions of unfolding the set of fixed vertebrae is the same.

% with
%\begin{align*}
%	&\sum_{i=1}^k(-1)^i\cdot n_i=n &&\text{and}
%	&&0\leq \sum_{i=1}^\ell(-1)^i\cdot n_i\leq n\text{ for all }\ell\in[k].
%\end{align*}

\begin{lemma}\label{lem:MultiUnfoldingCaterpillar}
	Let $\bC$ be a caterpillar and $\bC'$ be the $(n_1,\dots,n_k)$-unfolding of $\bC$ fixing $v_1,\dots,v_n$, then $\bC'$ is isomorphic to an unfolding of $\bC$ fixing $v_1,\dots,v_n$.
\end{lemma}
\begin{proof}
	Let $(a_0,\dots,a_m)$ be the spine of $\bC$. % and $(a'_0,\dots,a'_{m'})$ the spine of $\bC'$. 
	If $k=0$, then $\bC'=\bC[0,m]=\bC$ is an unfolding of $\bC$ fixing $v_1,\dots,v_n$.
	If $k>0$, then let $d_0=n_1, d_1=|n_2-n_1|,\dots,d_{k-1}=|n_k-n_{k-1}|, d_k=m-n_k$. Note that if there is an $\ell\in[1,k-1]$ with $d_{\ell-1}\geq d_\ell\leq d_{\ell+1}$, then $\bC'$ is isomorphic to the $(a''_{d_0+\dots+d_{\ell-1}-d_{\ell}},a''_{d_0+\dots+d_{\ell-1}})$-unfolding of $\bC''$, where $\bC''$ is the $(n_1,\dots,n_{\ell-1},n_{\ell+2},\dots,n_k)$-unfolding of $\bC$ and $(a''_0,\dots,a''_{m''})$ the spine of $\bC''$. Clearly, both unfoldings fix $v_1,\dots,v_n$. By induction hypothesis $\bC''$ is isomorphic to an unfolding of $\bC$ fixing $v_1,\dots,v_n$. Hence, by definition of unfoldings, $\bC'$ is isomorphic to an unfolding of $\bC$ fixing $v_1,\dots,v_n$.
	It remains to be shown that such an $\ell$ exists. Assume that there is no such $\ell$ in $[1,k-1]$. Then $d_\ell>\min(d_{\ell-1},d_{\ell+1})$ for all $\ell\in[1,k-1]$. Recall that $n_2\leq n_1$, hence, $d_0=n_1\geq d_1>d_2$. 
	Now $d_2>\min(d_{1},d_{3})$ implies $d_2>d_3$. Analogously, $d_3>\dots> d_k$. 
	%Note that $n_\ell=d_1-d_2+d_3-\dots -(-1)^\ell \cdot d_\ell$. %The fact that the $d_i$ are decreasing implies $n_{2\ell+1}=d_1-d_2+d_3-\dots -d_{2\ell}+d_{2\ell+1}\leq d_1-d_2+d_3-\dots -d_{2\ell-2}+d_{2\ell-1}=n_{2(\ell-1)+1}\leq\dots\leq n_1$. 
	Since $k$ is even we have $n_k<n_{k-1}$. Therefore $d_k=m-n_k\geq n_{k-1}-n_k=|n_{k-1}-n_k|=d_{k-1}$, a contradiction.
\end{proof}

The converse is also true but we do not need it for the following proofs.
%\begin{lemma}\label{lem:MultiUnfoldingCaterpillarOtherDirection}
%	Let $\bC$ be a caterpillar and $\bC'$ be an unfolding of $\bC$ fixing $v_1,\dots,v_n$, then there exist $n_1,\dots,n_k\in\N$, such that $\bC'$ is isomorphic to the $(n_1,\dots,n_k)$-unfolding of $\bC$ fixing $v_1,\dots,v_n$.
%\end{lemma}
%\begin{proof}
%	easy to see (i guess we have to do an induction)
%\end{proof}
%----------------------------------------------------------------------------------
Our goal in this section will be to show  that for a structure $\bB$ with caterpillar duality such that $\overline{\Csp(\bB)}$ is closed under $n$-fixed unfoldings of caterpillars we have that the canonical lam Datalog program $\Pi$ for $\bB$ and the canonical s$_n$lam Datalog program $\Pi_S$ for $\bB$ derive the goal predicate on the same instances. 
Note that when $\Pi$ is deriving the goal predicate we can use the `strongest possible rules', i.e., if we have different rules with the same body we always use the one with the smallest (with respect to inclusion) IDB in the head. However, $\Pi_S$ might need to use weaker rules in order to be able to apply symmetric rules later on in the derivation. See Example~3.6 in \cite{bodirskyStarke2024symmetriclineararcmonadicUnfoldedCaterpillar} for a derivation where this is necessary. 
Before we come the the proof we take a look at the close link of derivations of canonical lam Datalog programs and caterpillars. 
\begin{observation}\label{obs:lamDerivationIsCaterpillar}
%Let $\B$ be a finite structure with caterpillar duality, let $\Pi$ be the canonical linear arc monadic Datalog program for a structure $\B$, let $\bA$ be a finite structure that does not have a homomorphism to $\B$. Since $\B$ has caterpillar duality Theorem~\ref{thm:mainOld} implies that $\Pi$ can derive the goal predicate on $\bA$. Let $R_0,\dots, R_{n+1}$ be the rules of $\Pi$ used in such a derivation. 
Let $\Pi$ be a linear arc monadic Datalog program that solves the CSP of a finite structure $\bB$ and $\bA$ be a structure on which $\Pi$ can derive the goal predicate.
Let $R_0,\dots, R_{n+1}$ be the rules of $\Pi$ used in such an efficient derivation. 
Since $\Pi$ is linear, we can use the suggestive notation 
    \[\vdash_{R_0} P_0(a_0)\vdash_{R_1} \cdots \vdash_{R_{n}}P_n(a_n)\vdash_{R_{n+1}}G\]
    for this derivation. By Lemmata~\ref{lem:DLinjective} and~\ref{lem:DLconnected} we can assume without loss of generality that all rules used in this derivation are injective and connected.
    The rule $R_0$ is of the form $P_0(x) \; {:}{-} \;\Psi_0(x_1,\dots,x_k)$ and $R_{n+1}$ is of the form $G \; {:}{-} \; P_n(x)\wedge\Psi_{n+1}(x_1,\dots,x_k)$, where $\Psi_0$ and $\Psi_{n+1}$ are atomic formulas. Let $\Phi_0(x)$ and $\Phi_{n+1}(x)$ be obtained from $\Psi_0$ and $\Psi_{n+1}$, respectively, by existentially quantifying all variables except for $x$. Note that, since the derivation is efficient, any other rule $R_i$ is either of the form $P_i(x) \; {:}{-} \;P_{i-1}(y)\wedge\Psi_i(x_1,\dots,x_k)$ with $x\neq y$ or $P_i(x) \; {:}{-} \;P_{i-1}(x)\wedge\Psi_i(x_1,\dots,x_k)$. Analogously to $\Phi_0(x)$ and $\Phi_{n+1}(x)$ define the formula $\Phi_i(x,y)$ or $\Phi_i(x)$.
    Let $0<i_0<\dots<i_{m-1}<{n+1}$ such that $\Phi_{i_0}, \dots,\Phi_{i_{m-1}}$ are exactly the formulas with two free variables. For $\ell\in[0,m-1]$ define $\phi_\ell(x,y)$ as $\Phi_{i_\ell}(x,y)$. For $\ell\in[0,m]$ define $\psi_\ell(x)$ as $\Phi_{i_{\ell-1}+1}(x)\wedge\dots\wedge\Phi_{i_\ell-1}(x)$, where $i_{-1}=-1$ and $i_m=n+2$. Note that $\psi_\ell$ can be the empty conjunction. Observe that \[\phi_0(x_1,x_2)\wedge\dots\wedge\phi_{m-1}(x_{m-1},x_m)\wedge\psi_0(x_0)\wedge\dots\wedge\psi_m(x_m)\] is the conjunctive query representation of a caterpillar $\bC$. Note that this construction is essentially a proof of (2)${}\Rightarrow{}$(1) in Theorem~\ref{thm:mainOld}. We will use the construction for something different. For $\ell\in [0,m-1]$ define $S_\ell$ as the rule $R_{i_\ell}$ and for $\ell\in[1,m-1]$ define $H_\ell$ as a rule with head $P_{i_{\ell+1}}(x)$ and body the canonical conjunctive query of the canonical database of $P_{i_\ell+1}(x)\wedge \psi_\ell(x)$. The rules $H_0$ and $H_m$ are defined analogously. Observe that adding the rules $H_0,\dots,H_m$ to $\Pi$ does not change the set of structures on which $\Pi$ can derive the goal predicate. Hence, we can without loss of generality say that
    \[\vdash_{H_0} \tilde P_0(\tilde a_{0})\vdash_{S_0} \tilde Q_1(\tilde a_{1})\vdash_{H_1}\tilde P_1(\tilde a_{1})\vdash_{S_1} \cdots \vdash_{H_{m-1}}\tilde P_{m-1}(\tilde a_{{m-1}})\vdash_{S_{m-1}}\tilde Q_m(\tilde a_{m})\vdash_{H_{m}}G\]
    is a derivation of $\Pi$ on $\bA$ that derives the goal predicate, where $\tilde P_\ell\coloneqq P_{i_\ell}$, $\tilde Q_\ell\coloneqq P_{i_\ell+1}$, $\tilde a_\ell\coloneqq a_{i_\ell}$, and $\tilde a_m\coloneqq a_{i_{m-1}+1}$. The advantage of the new derivation is that it is easier to translate between the derivation and the corresponding caterpillar $\bC$.
\end{observation}

Now we are ready to prove that we can transform derivations of $\Pi$ into derivations of $\Pi_S$. 
\begin{lemma}\label{lem:MaltimpliesLAMDatalogIsSLAMDatalog}
Let $\bB$ be a finite structure with relational signature $\tau$ such that $\bB$ has caterpillar duality and $\overline{\Csp(\bB)}$ is closed under $n$-fixed unfoldings of caterpillars. %$\Pol(\bB)$ contains a quasi Maltsev operation. 
Let $\Pi$ be the canonical linear arc monadic Datalog program for $\bB$ and $\Pi_S$ be the canonical s$_n$lam Datalog program for $\bB$. Then $\Pi$ can derive the goal predicate on a finite 
$\tau$-structure 
$\bA$ if and only if $\Pi_S$ can derive the goal predicate on $\bA$.
\end{lemma}
The proof is a generalization of the proof of Lemma~3.4 in~\cite{bodirskyStarke2024symmetriclineararcmonadicUnfoldedCaterpillar}.
\begin{proof}
    Note that every rule of $\Pi_S$ is also a rule of $\Pi$. Hence if $\Pi_S$ can derive the goal predicate on $\bA$, then so can $\Pi$.
    Consider a derivation of the goal predicate $G$ for the Datalog program $\Pi$ on $\bA$. Let $R_0,\dots, R_{m+1}$ be the rules of $\Pi$ used in this derivation. We assume that the derivation has the form discussed in Observation~\ref{obs:lamDerivationIsCaterpillar}
    \[\vdash_{H_0} P_0(a_0)\vdash_{S_0} Q_1(a_1)\vdash_{H_1}P_1(a_1)\vdash_{S_1} \cdots \vdash_{H_{m-1}}P_{m-1}(a_{m-1})\vdash_{S_{m-1}}Q_m(a_m)\vdash_{H_{m}}G\]
    and that $\phi_0(x_1,x_2),\dots,\phi_{m-1}(x_{m-1},x_m)$ and $\psi_0(x_0),\dots,\psi_m(x_m)$ are the formulas corresponding to $S_0,\dots,S_{m-1}$ and $H_0,\dots,H_m$, respectively.    
    %For $i\in[n]$ define the binary relation $\to_i$ on $B$ that contains all tuples $(b,b')$ such that $\bB\models \Phi_i(b,b')$.
    Since $\Pi$ is canonical we may assume without loss of generality that all rules strongest possible, i.e.,
    \begin{itemize}
        \item[(C1)] $P_0=\{b\in B\mid \bB\models \Phi_0(b)\} =\Phi_0^{\bB}$,
        \item[(C2)] $P_{i}=\{b\in Q_{i}\mid \bB\models \psi_{i}(b)\}$ for all $i\in[m-1]$, and 
        \item[(C3)] $Q_i=\{b'\in B\mid \exists b\in P_{i-1}\colon \bB\models\phi_{i-1}(b,b')\}$ for all $i\in[m]$.
    \end{itemize}

    Note that $\phi_0(x_1,x_2)\wedge\dots\wedge\phi_{m-1}(x_{m-1},x_m)\wedge\psi_0(x_0)\wedge\dots\wedge\psi_m(x_m)$ is the conjunctive query representation of a caterpillar $\bC$ with spine $(x_0,\dots,x_m)$. Replacing $a_i$ by $x_i$ in the above derivation we obtain a derivation of $\Pi$ that derives the goal predicate on $\bC$. Therefore $\bC\not\to\bB$. By assumption there are vertebrae $v_1,\dots,v_n$ such that no unfolding of $\bC$ that fixes $v_1,\dots,v_n$ has a homomorphism to $\bB$.

    Let $\tilde Q_1,\dots,\tilde Q_m,\tilde P_0,\dots,\tilde P_{m-1}\subseteq B$ be the smallest sets such that
    \begin{enumerate}
        \item[(P)] $P_i\subseteq \tilde P_i$ and $Q_{i+1}\subseteq \tilde Q_{i+1}$ for all $i\in[m-1]$,
        \item[(H1)] $b\in \tilde Q_i$ and $\bB\models\psi_i(b)$ implies $b\in \tilde P_i$ for all $i\in[m-1]$,
        \item[(H2)] $b\in \tilde P_i$ and $\bB\models\psi_i(b)$ implies $b\in \tilde Q_i$ for all $i\in[m-1]$,
        \item[(S1)] $b\in \tilde P_i$ and $\bB\models\phi_i(b,b')$ implies $b'\in \tilde Q_{i+1}$ for all $i\in[0,m-1]$, and
        \item[(S2)] $b'\in \tilde Q_{i+1}$ and $\bB\models\phi_i(b,b')$ implies $b\in \tilde P_{i}$ for all $i\in[0,m-1]$ with $(x_i,x_{i+1})\notin\{v_1,\dots,v_n\}$.
    \end{enumerate}
    Note that there can be $b\in Q_i$ such that there is no $b'\in B$ with $b\to_i b'$. Analogously, there can be $b'\in Q_{i+1}$  such that there is no $b\in B$ with $b\to_i b'$.
    
    Let $\tilde H_0,\dots,\tilde H_{m},\tilde S_0,\dots,\tilde S_{m-1}$ be the rules obtained from $ H_0,\dots, H_{m},S_0,\dots,S_{m-1}$ by replacing each occurrence of $P_i$ and $Q_i$ by $\tilde P_i$ and $\tilde Q_i$, respectively.
    We will now show that 
    \[\vdash_{\tilde H_0} \tilde P_0(a_0)\vdash_{\tilde S_0} \tilde Q_1(a_1)\vdash_{\tilde H_1}\tilde P_1(a_1)\vdash_{\tilde S_1} \cdots \vdash_{\tilde H_{m-1}}\tilde P_{m-1}(a_{m-1})\vdash_{\tilde S_{m-1}}\tilde Q_m(a_m)\vdash_{\tilde H_{m}}G\]
    is a derivation of $\Pi_S$ that derives the goal predicate on $\bA$. Recall that $\Pi_S$ consists of $n+1$ many copies of the canonical slam Datalog program that are connected by rules of the canonical lam Datalog program. Hence, strictly speaking, renaming the IDBs is necessary to obtain rules in $\Pi_S$. We will omit giving the renaming explicitly and instead show that $\tilde H_0,\dots,\tilde H_{m},\tilde S_0,\dots,\tilde S_{m-1}$ are rules of $\Pi$ (the canonical lam Datalog program for $\bB$) and that all but at most $n$ of those rules are symmetric. 
    By definition $P_0\subseteq \tilde P_0$. Hence $\bB\models\Phi_0(b)$ implies $b\in\tilde P_0$ and $\tilde H_0$ is a rule of $\Pi$. Since there is no IDB in its body the rule is also symmetric. 
    Let $i\in[0,m-1]$. By (S1) $\tilde S_i$ is a rule of $\Pi$. 
    If $(x_i,x_{i+1})\notin\{v_1,\dots,v_n\}$, then, by (S2), the reversed rule of $S_i$ is also a rule of $\Pi$. Hence all but at most $n$ of the rules $\tilde S_0,\dots,\tilde S_{m-1}$ are symmetric rules of $\Pi$.
    Let $i\in[m-1]$. By (H1) and (H2) $\tilde H_{i}$ and its reverse are rules of $\Pi$.
    Finally we need to show that $\tilde H_m$ is a symmetric rule of $\Pi$. Note that since it derives the goal predicate it automatically is a symmetric rule. We need to show that there is no $b\in \tilde Q_m$ such that $\B\models \psi_m(b)$. Assume that such a $b$ exists. Then $b\in \tilde Q_m$ because of (P) or because of (S1). By (C3), in both cases there is a $b_M\in \tilde P_{m-1}$ such that $\phi_{m-1}(b_M,b)$. The element $b_M$ was added to $\tilde P_{m-1}$ because of 
    \begin{itemize}
        \item (P) in which case, by (C2), $\B\models \psi_{m-1}(b_M)$,
        \item (H1) in which case $\B\models \psi_{m-1}(b_M)$, or
        \item (S2) in which case there exists $b_{M-1}\in Q_m$ with $\B\models\phi_{m-1}(b_{M-1},b_{M})$.
    \end{itemize}
    Continuing in this manner we can construct a path $b_M,\dots,b_0$ to an element $b_0\in\tilde P_0$. There is a unique map $\iota\colon[0,M]\to[0,m]$ such that $b_\ell\in P_{\iota(\ell)}$ or $b_\ell\in Q_{\iota(\ell)}$ for all $\ell$. Let $0<n_1<\dots<n_k<M$ be all numbers $\ell$ such that the sequence  $(\iota(\ell-1),\iota(\ell),\iota(\ell+1))$ is not monotone. Note that $\iota(n_1),\dots,\iota(n_k)$ mark the orientation changes of the path $b_0,\dots,b_M,b$. This path  yields an $(\iota(n_1),\dots,\iota(n_k))$-unfolding $\bC'$ of $\bC$. Observe that, by (S2), $\bC'$ fixes $v_1,\dots,v_n$. By Lemma~\ref{lem:MultiUnfoldingCaterpillar} $\bC'$ is an unfolding of $\bC$ that fixes $v_1,\dots,v_n$ with spine $b_0,\dots,b_M,b$ such that each $b_j$ is an element of the corresponding $\tilde P_i$ and/or $\tilde Q_i$. By construction the map $\bC'\to\B,b_j\mapsto b_j,b\mapsto b$ can be extended to a homomorphism. A contradiction. Therefore such a $b$ can not exist and $\tilde H_m$ is a rule of $\Pi$. 
    In conclusion, $\Pi_S$ can derive the goal predicate on $\bA$ as  desired.
\end{proof}

%-----------------------------------------------------
\section{Going back from Datalog to Caterpillars}%S$_n$LAM Datalog implies unfolded caterpillar duality}

Let $\Pi$ be a lam Datalog program, $\bC$ be a caterpillar, $v=(a,b)$ be a vertebra of $\bC$, and let $t\in R^{\bC}$ be the unique $R$ and $t$ for which $a,b$ occur in $t$. A derivation of $\Pi$ applies non-symmetric rule to $v$ if the derivation contains a rule $R$ that can only be applies because its EDB matches $R(t)$ and the reversed rule $\tilde R$ is not a rule of $\Pi$.

\begin{lemma}\label{lem:lamDatalogUnfoldings}
    Let $\Pi$ be a lam Datalog program and $\bC$ be a caterpillar. If $\Pi$ can derive the goal predicate on $\bC$ and $v_1,\dots,v_n$ are the vertebrae of $\bC$ to which this derivation applied non-symmetric rules, then $\Pi$ can derive the goal predicate on any unfolding of $\bC$ that fixes $v_1,\dots,v_n$.
\end{lemma}
\begin{proof}
    Let $(a_0,\dots,a_m)$ be the spine of $\bC$ and $\phi\wedge\psi$ be the conjunctive query representation of $\bC$. Consider a derivation of $\Pi$ that derives the goal predicate on $\bC$. Let $R_0,\dots,R_{k+1}$ be the rules of $\Pi$ used in this derivation. Since $\Pi$ is linear and monadic, we can use the notation
    \begin{align*}
        \vdash_{R_0}P_0(c_0)\vdash_{R_1}\cdots\vdash_{R_k}P_k(c_k)\vdash_{R_{k+1}}G
    \end{align*}
    for this derivation. 
    Let $\bC'$ be an $(a,b)$-unfolding of $\bC$ such that no $v_1,\dots,v_n$ lies between $a$ and $b$. 
    We will show that $\Pi$ can derive the goal predicate on $\bC'$ without applying non-symmetric rules to $v_1,\dots,v_n$.
    The lemma then follows by induction.
    %By induction it suffices to show the statement for the case that $\bC'$ is an $(a_i,a_j)$-unfolding of $\bC$ that fixes $v_1,\dots,v_n$. 

    Note that if $P_\ell(c_\ell)\vdash_{R_{\ell+1}}P_{\ell+1}(c_{\ell+1})\vdash_{R_{\ell+2}}P_{\ell+2}(c_{\ell+2})$ and there is no relation in $\bC$ that contains a tuple which contains both $c_\ell$ and $c_{\ell+2}$, then $c_{\ell+1}$ is a joint. Hence the derivation can not jump over joints, i.e., for all $i<j$ and all joints $c$ between $c_i$ and $c_{j}$ there is a $\ell\in[i,j]$ with $c_\ell=c$.
    Consider an interval
    \begin{align*}
        P_i(c_i)\vdash_{R_{i+1}}\cdots\vdash_{R_j}P_j(c_j)
    \end{align*}
    of the derivation on $\bC$ such that $\{c_{i},c_j\}=\{a,b\}$ and $c_{i+1},\dots,c_{j-1}\notin\{a,b\}$ 
    If $c_i=a$ and $c_j=b$, then replace this interval with
    \begin{align*}
P_i(a)&\vdash_{R_{i+1}}P_{i+1}(c'_{i+1})\vdash_{R_{i+2}}\cdots\vdash_{R_{j-1}} P_{j-1}(c'_{j-1})\vdash_{R_j}P_j(b')\\
%\vdash_{\tilde R_{j-1}} P_{j-1}(c_{j-1})
P_i(a'')&\dashv_{\tilde R_{i+1}}P_{i+1}(c_{i+1})\dashv_{\tilde R_{i+2}}\cdots\dashv_{\tilde R_{j-1}} P_{j-1}(c_{j-1})\dashv_{\tilde R_j}P_j(b')\\
P_i(a'')&\vdash_{R_{i+1}}P_{i+1}(c''_{i+1})\vdash_{R_{i+2}}\cdots\vdash_{R_{j-1}} P_{j-1}(c''_{j-1})\vdash_{R_j}P_j(b)
    \end{align*}
    where $\tilde R_\ell$ is the reversed rule of $R_\ell$ and $c'$ and $c''$ come from the naming convention fixed in Definition~\ref{def:unfolding}.
    Since no $v_1,\dots,v_n$ lies between $a$ and $b$ the rules $R_{i+1},\dots, R_{j-1}$ are symmetric. Hence $\tilde R_{i+1},\dots,\tilde R_{j-1}$ are also rules of $\Pi$.
    If $c_i=b$ and $c_j=a$, then we replace the interval analogously. The resulting derivation is a derivation of $\Pi$ on $\bC'$ that derives the goal predicate and never applies a non-symmetric rule to $v_1,\dots,v_n$.
\end{proof}

\begin{lemma}\label{lem:SLAMimpliesUnfoldedCaterpillar}
    If $\Csp(\bB)$ is solved by a s$_n$lam Datalog program $\Pi$, then $\bB$ has  caterpillar duality and $\overline{\Csp(\bB)}$ is closed under $n$-fixed unfoldings of caterpillars. 
\end{lemma}
\begin{proof}
    The Datalog program $\Pi$ is, in particular, linear arc monadic. Hence, by Theorem~\ref{thm:mainOld}, $\bB$ has caterpillar duality.
    Let $\bC$ be a caterpillar that has no homomorphism to $\bB$. Since $\Pi$ solves $\Csp(\bB)$, there exists a derivation of $\Pi$ on $\bC$ that derives the goal predicate. Let $v_1,\dots,v_{m}$ be the vertebrae of $\bC$ to which non-symmetric rules are applied.  By Lemma~\ref{lem:lamDatalogUnfoldings}, every unfolding of $\bC$ that fixes $v_1,\dots,v_m$ is also an element of $\overline{\Csp(\bB)}$. Since $\Pi$ is a s$_n$lam Datalog program, we have that $m\leq n$.   
    Therefore, $\overline{\Csp(\bB)}$ is closed under $n$-fixed unfoldings of caterpillars.
\end{proof}

\section{From Caterpillars to Hagemann-Mitschke Chains}
In this section we prove $(\ref{caterpillar})\Rightarrow(\ref{non-degenerate-minor})$ from Theorem~\ref{thm:main}. 
\begin{lemma}\label{lem:unfoldedCPDualityImpliesSigmaImpliesHM}
    Let $n\geq0$ and let $\bB$ be a relational structure with $n$-fixed  unfolded caterpillar duality. If $\Pol(\bB)$ does not satisfy a minor condition $\Sigma$, then $\Sigma$ implies (for finite structures) $\HM n$. 
\end{lemma}

The proof is a generalization of the proof of Lemma~3.10 in~\cite{bodirskyStarke2024symmetriclineararcmonadicUnfoldedCaterpillar}. The rough idea of the proof is the following:
\begin{enumerate}
	\item Consider the indicator structure $\bI$ of $\Sigma$ with respect to $\bB$. 
	\item Since $\bB\not\models\Sigma$, there must be a forbidden caterpillar $\bC$ that has a homomorphism $h$ into $\bI$.
	\item Find the Hagemann-Mitschke chain in the image of $h$. 
\end{enumerate} 
\begin{proof} 
Let $\tau$ be the signature of $\bB$. 
Suppose that 
$\Pol(\bB)\not\models\Sigma$. 
By Lemma~\ref{lem:WlogSingleFunctionSymbol}, we may assume that $\Sigma$ only involves a single function symbol $f$ of arity $K$. 
Let $\bI$ be the indicator structure of $\Sigma$ with respect to $\bB$. Then $\bI\not\to\bB$ (Lemma~\ref{lem:ind}). 
By definition $\bB$ has caterpillar duality, hence there must be a caterpillar $\bC$ with a homomorphism to $\bI$ that does not have a homomorphism to $\bB$. 
Let $(s_0,\dots,s_m)$ be the spine of $\bC$. By $n$-fixed unfolded caterpillar duality there are vertebrae $v_1,\dots,v_n$ of $\bC$ such that no unfolding of $\bC$ that fixes $v_1,\dots,v_n$ has a homomorphism to $\bB$.
In the following we will show that  
\begin{itemize}
    \item $\Sigma$ implies $\HM n$, or that, 
    \item there exists an unfolding of $\bC$ fixing $v_1,\dots,v_n$ that has a homomorphism to $\bB$, which contradicts the $n$-fixed unfolded caterpillar duality of $\bB$.
\end{itemize} 
Showing this concludes the proof.

discuss case $m=0$ and $\psi_0$ only has 1 conjunct

Let $h$ be the maximal number of hairs of any $s_i$, i.e. the maximal number of conjuncts in any $\psi_i$. 

Let $\alpha$ be a homomorphism from $\bC$ to $\bI$. Note that the element $\alpha(s_i)$ of the indicator structure $\bI$ is a set of $K$-tuples of elements in $B$. To be able to better work with $\alpha$ we will choose representatives of these sets. We pick them in such a way that they witness that $\alpha$ is a homomorphism. %See Figure~\ref{fig:explainTiNotation} for a visualization of the chosen tuples in an example.   
Formally, fix tuples $t_{0}^0,\dots,t_{0}^{h+1},t_{1}^0,\dots,t_{1}^{h+1},\dots,t_{m}^0,\dots,t_{m}^{h+1}$ in $\bB^K$ such that
\begin{enumerate}
\item[(P1)] {for every $c \in [0,m]$ we have $t_{c}^0,\dots,t_{c}^{h+1}\in \alpha(s_c)\subseteq B^K$,}
    \item[(P2)] for every $c \in [0,m]$ and every conjunct $\psi(s_c)$ of $\psi_c(s_c)$ there is an $i\in[h]$ such that {$\bB^K\models\psi(t_c^i)$} %{, here $t_{i,j,k}$ denotes the $k$-th entry of the tuple $t_{i,j}$,}
    and
    %\item for every $i\in [n]$ we have that $\bB\models \phi_i(t_{i-1,m+1,k},t_{i,0,k})$ for all $k\in[K]$. 
    \item[(P3)] {for every $c\in [0,m-1]$ we have  $\bB^K\models\phi_c(t_{c}^{h+1},t_{c+1}^0)$.} 
\end{enumerate}
{
The existence of such tuples follows directly from $\alpha$ being a homomorphism and $\bI$ being an indicator structure. 
Note that   $\bB^K\models\psi(t_c^i)$ is equivalent to $\bB\models \psi((t_c^i)_k)$ for all $k\in[K]$ and that, similarly,  $\bB^K\models\phi_c(t_{c}^{h+1},t_{c+1}^0)$ is equivalent to $\bB\models \phi_c((t_{c}^{h+1})_k,(t_{c+1}^0)_k)$ for all $k\in[K]$.  
} 
%Such tuples exist directly by virtue of $h$ being a homomorphism and $\bI$ being an indicator structure.}

\begin{figure}
    \centering
        \begin{tikzpicture}[scale=0.55]
        	\node at (0,-1) {\figureCaterpillar};
        	
        	\node at (0,-5) {\figureBlocks};
        	
        	%\node[opacity=0.5] at (0,-10) {\figureBlocks};
        	%todo opacity version with x any instead of numbers?
        	%\node at (0,-10) {\figureBraidOverlay};
        	\node at (0,-10) {\figureBraid};
        	%\node at (0,-10) {\figureBlocksBraids};
        	
    \end{tikzpicture}
    \caption{
    At the top there is a picture of two fixed vertebrae of a caterpillar $\bC$.
    Below there is a visualization of the map $\pi$ and the relations $\connect{0}{c}$, $\protect\toEdge_{i_c}$,  $\connect{c}{c+1}$, $\protect\toEdge_{i_{c+1}}$,  $\connect{c+1}{c+2}$, $\protect\toEdge_{i_{c+2}}$, and  $\connect{c+2}{c+3}$. Observe that there is a 3-braid embedded in this picture.
    At the bottom we see that this construction yields the identities $p_c(x,x,y)\approx p_{c+1}(x,y,y)$ and $p_{c+1}(x,x,y)\approx p_{c+2}(x,y,y)$.}
    \label{fig:fixedEdgesToHM}
\end{figure}
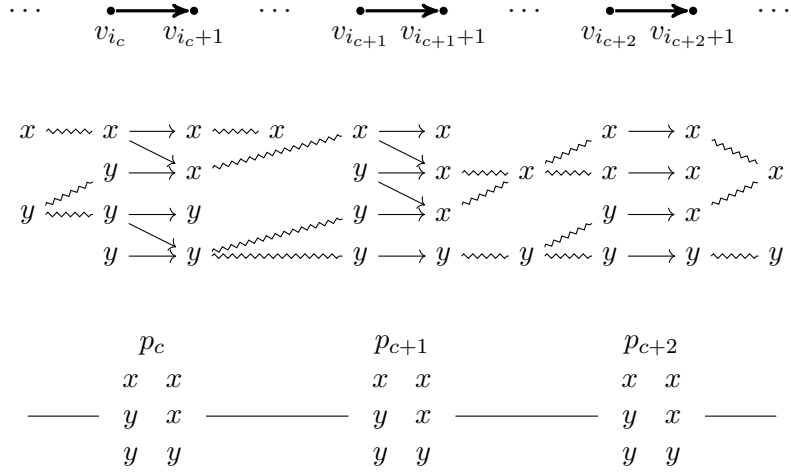

Define binary relations $\sim$ and $\toEdge_a$ on $[0,m]\times [0,h+1]\times [K]$ for $c\in[0,m-1]$: %for $(c,i,k),(d,j,\ell)\in [n]\times [0,m+1]\times [K]$  we have 
\begin{itemize}
    %\item $(c,i,k)\sim (d,j,\ell)$ if and only if {$c=d$ and $(t_c^i)_k=(t_d^j)_\ell$, in other words, the tuples $t_c^i$ and $t_d^j$ represent the same element in $\bI$ and the element of $\bB$ that is the $k$-th entry of $t_c^i$ is also the $\ell$-th entry of $t_d^j$}
    %\item $(c,i,k) \toEdge_a (d,j,\ell)$  if and only if $c+1=d=a$, $i=m+1$, $j=0$, and {$\bB\models\phi_d((t_c^i)_k,(t_d^j)_\ell)$}.
    \item $(c,i,k)\sim (c,i+1,\ell)$ if and only if {$(t_c^i)_k=(t_c^{i+1})_\ell$, note that the tuples $t_c^i$ and $t_c^{i+1}$ represent the same element in $\bI$}
    \item $(c,h+1,k) \toEdge_c (c+1,0,\ell)$  if and only if {$\bB\models\phi_c((t_c^{m+1})_k,(t_{c+1}^0)_\ell)$}, {note that because of (P3) this conditions holds whenever $k=\ell$, though it may also hold for some $k\neq\ell$.}
\end{itemize}
We write $\fromEdge_c$ for the converse of $\toEdge_c$. 
For $c\in [n]$ the vertebra $v_c$ is of the form $(s_{i_c},s_{i_c+1})$. Define $i_0=-1$ and $i_{n+1}=m$. 
For $c\in [0,n]$ we define the binary relation $\connect{c}{c+1}$ on $[0,m]\times [0,h+1]\times [K]$ as the reflexive transitive  symmetric closure of $\sim\cup \toEdge_{i_c+1}\cup\dots\cup\toEdge_{i_{c+1}-1}$.
Now we consider two cases. The first case
is that there are no $k,\ell\in[K]$ such that 
\[(0,0,k) \connect{0}{1}{}\circ{} \toEdge_{i_1}{}\circ{} \connect{1}{2}{}\circ{} \toEdge_{i_2} {}\circ{}\dots{}\circ{} \connect{n-1}{n}{}\circ{} \toEdge_{i_n}   {}\circ{} \connect{n}{n+1} (m,h+1,\ell).\] 

%$H_c=\{k\in[K]\mid \exists\ell\in[K]\ (0,0,\ell) {}_0\to_{i_c} (i_c,h+1,k)\}$

%$T_c=\{k\in[K]\mid \exists\ell\in[K]\ (i_c+1,0,k) {}_{c}\to_{n+1} (m,h+1,\ell)\}$

%$M_c=\{k\in[K]\mid \exists \ell\in T_c\ (i_c,h+1,k) \connect{c-1}{c}\circ{}\toEdge_{i_c}(i_{c}+1,0,\ell)\}\cup \{k\in[K]\mid \exists \ell\in H_c\ (i_c,h+1,\ell) \connect{c}{c+1}{}\circ{}\toEdge_{i_c}(i_{c}+1,0,k)\}$

%Note that by assumption $T_c\cap H_c=T_c\cap M_c=M_c\cap H_c=\emptyset$.

Define the map $\pi$ from $[0,m]\times [0,h+1]\times [K]$ to $\{x,y\}$ with $\pi(c,i,\ell)=x$ if there exists $\ell\in[K]$ and $d\in[0,n]$ such that 
\[(0,0,k) \connect{0}{1}{}\circ{} \toEdge_{i_1}{}\circ{} \connect{1}{2}{}\circ{} \dots{}\circ{} \toEdge_{i_d}   {}\circ{} \connect{d}{d+1} (c,i,\ell)\]
and $\pi(c,i,\ell)=y$ otherwise.

See Figure~\ref{fig:fixedEdgesToHM} for a visualization. Note the similarities with Figure~\ref{fig:HM3}.
Observe the following.
\begin{enumerate}
    \item {By definition} $\pi((0,0,k))=x$ for all $k\in[K]$.
    \item By assumption $\pi((m,h+1,\ell))=y$ for all $\ell\in[K]$.
    \item Let $c\in[0,m]$ and let $i\in[0,h]$. {By (P1) both $t_c^i$ and $t_c^{i+1}$ represent the element $s_c$ in the indicator structure $\bI$. Hence, by definition of the indicator structure, when viewing the entries of $t_c^i$ and $t_c^{i+1}$ as variables, $\Sigma$ implies $f(t_c^i)\approx f(t_c^{i+1})$. By definition of $\pi$ and $\sim$ we have that \[\text{$(t_c^{i})_{k_1}=(t_c^{i+1})_{k_2}$ implies $\pi(c,i,{k_1})=\pi(c,i+1,{k_2})$ for all $k_1,k_2\in[K]$.}\] Therefore the minor condition $f(t_{c,i})\approx f(t_{c,i+1})$ implies the minor condition} 
    \[f(\pi(c,i,1),\dots,\pi(c,i,K))\approx f(\pi(c,i+1,1),\dots,\pi(c,i+1,K)).\]
    \item Let $c\in[0,m-1]$.  If $c\notin \{i_{1},\dots i_n\}$, then, {by definition of $\toEdge_c$ and (P3), we have that}  $\pi(c,h+1,k)=\pi(c+1,0,k)$ for all $k\in[K]$. Hence, 
    \[f(\pi(c,h+1,1),\dots,\pi(c,h+1,K))= f(\pi(c+1,0,1),\dots,\pi(c+1,0,K)).\]
    If $c\in \{i_{1},\dots i_n\}$, then, {by definition of $\toEdge_{c}$ and (P3), we have that}  $\pi(c,h+1,k)=x$ implies $\pi(c+1,0,k)=x$ for all $k\in[K]$. %Hence, 
    %\[f(\pi(c,m+1,1),\dots,\pi(c,m+1,K))= f(\pi(c+1,0,1),\dots,\pi(c+1,0,K)).\]
\end{enumerate}

Define the height-1 condition $\Sigma'$ as the set consisting of all identities from $\Sigma$ and the identities 
\begin{itemize}
    \item $p_0(x)\approx f(x,\dots,x)$, 
\item for $c\in[n]$ \[p_c(x,y,z)\approx f(\tau(\sigma(c,1)),\dots, \tau(\sigma(c,K))\] 
where $\sigma(c,j)=(\pi(i_c,h+1,j),\pi(i_c+1,0,j))$ and $\tau(x,x)= x, \tau(y,x)= y$, and $\tau(y,y)= z$, and 
\item $p_{n+1}(y)\approx f(y,\dots,y)$.
\end{itemize}
Since each $p_c$ only occurs in one identity of $\Sigma'$ and in this identity all variables also occur as parameters of $p_c$ we have that $\Sigma'$ and $\Sigma$ are equivalent. Note that $p_c(x,y,z)=p_c(\tau(x,x),\tau(y,x),\tau(y,y))$ and we could have also used $(x,x), (y,x)$, and $(y,y)$ as variable names, $\tau$ is just a renaming for better readability.

Define $\pi_1(x_1,x_2)=x_1$ and $\pi_2(x_1,x_2)=x_2$. Then, $\Sigma'$ (and hence $\Sigma$) implies 
\begin{align*}
p_0(x)&\approx f(x,\dots,x)\\ &=f(\pi(0,0,1),\dots,\pi(0,0,K))\\
& \approx f(\pi(0,1,1),\dots,\pi(0,1,K))\\
& \quad \vdots\\
%& \approx f(\pi(i_1,h+1,1),\dots,\pi(i_1,h+1,K))\approx p_1(x,y,y)
& \approx f(\pi(i_1,h+1,1),\dots,\pi(i_1,h+1,K))\\
& = f(\pi_1(\sigma(1,1)),\dots,\pi_1(\sigma(1,K)))\\
&\approx p_1(\pi_1(x,x),\pi_1(y,x),\pi_1(y,y))\\
&=p_1(x,y,y)
\intertext{for $c\in[n-1]$}
p_c(x,x,y)
&=p_c(\pi_2(x,x),\pi_2(y,x),\pi_2(y,y))\\
&\approx f(\pi_2(\sigma(c,1)),\dots,\pi_2(\sigma(c,K)))\\
&= f(\pi(i_c+1,0,1),\dots,\pi(i_c+1,0,K))\\
& \approx f(\pi(i_c+1,1,1),\dots,\pi(i_c+1,1,K))\\
& \quad \vdots\\
& \approx f(\pi(i_{c+1},h+1,1),\dots,\pi(i_{c+1},h+1,K))\\
& = f(\pi_1(\sigma(c+1,1)),\dots,\pi_1(\sigma(c+1,K)))\\
&\approx p_{c+1}(\pi_1(x,x),\pi_1(y,x),\pi_1(y,y))\\
&=p_{c+1}(x,y,y)
\intertext{and}
p_n(x,x,y)
&=p_n(\pi_2(x,x),\pi_2(y,x),\pi_2(y,y))\\
&\approx f(\pi_2(\sigma(n,1)),\dots,\pi_2(\sigma(n,K)))\\
&= f(\pi(i_n+1,0,1),\dots,\pi(i_n+1,0,K))\\
& \approx f(\pi(i_n+1,1,1),\dots,\pi(i_n+1,1,K))\\
& \quad \vdots\\
& \approx f(\pi(m,h+1,1),\dots,\pi(m,h+1,K))\\
&=f(y,\dots,y)\\
&\approx p_{n+1}(y).
\end{align*}
Hence, $\Sigma$ implies $\HM n$, as desired. The case $n=0$ is slightly different, for a detailed discussion of this case see proof of Lemma~3.10 in~\cite{bodirskyStarke2024symmetriclineararcmonadicUnfoldedCaterpillar}. 

Now we consider the case that there are $k,\ell\in[K]$ such that 
\[(0,0,k) \connect{0}{1}{}\circ{} \toEdge_{i_1}{}\circ{} \connect{1}{2}{}\circ{} \toEdge_{i_2} {}\circ{}\dots{}\circ{} \connect{n-1}{n}{}\circ{} \toEdge_{i_n}   {}\circ{} \connect{n}{n+1} (m,h+1,\ell).\] 
Then there exists a sequence  $a_0,\dots,a_M$ of distinct elements in $[0,m]\times[0,h+1]\times[K]$ such that $a_0=(0,0,k)$, $a_M=(m,h+1,\ell)$, and for every $i\in[M]$ we have $a_{i-1}\sim a_i$, $a_{i-1}\toEdge_j a_i$ for some $j\in[m]$, or $a_{i-1}\fromEdge_j a_i$ for some $j\in[m]\setminus\{i_1,\dots,i_n\}$.  
We will use this sequence to construct an unfolding of $\bC$ that fixes $v_1,\dots,v_n$ and has a homomorphism to $\bB${, contradicting that $\bB$ has $n$-fixed unfolded caterpillar duality while $\bC\not\to\bB$.} See Figure~4 in \cite{bodirskyStarke2024symmetriclineararcmonadicUnfoldedCaterpillar} 
for a visualization of the construction of the unfolding of $\bC$ that follows.  
First partition the sequence $0,\dots,M$ into the sequence $A_0,\dots,A_{M'}$ such that 
\begin{itemize}
    \item $a_0\in A_0$ and $a_M\in A_{M'}$,
    \item {each $A_i$ consists of an interval, i.e., for all $i\in[0,M']$ there are $0\leq j_1\leq j_2\leq M$ such that $A_i=\{{j_1},{j_1+1},\dots,{j_2}\}$},
    \item for all $i\in[0,M']$ and $j\in A_i$  with $j+1\in A_i$ we have $a_{j}\sim a_{j+1}$,
    \item for all $0\leq i_1<i_2\leq M'$ and all ${j_1}\in A_{i_1}$, ${j_2}\in A_{i_2}$ we have $j_1<j_2$, and
    \item for all $i\in[M']$ we have $\max(A_{i-1})+1=\min(A_i)$ and there is a $c\in[m]$ with $a_{\max(A_{i-1})}\toEdge_c a_{\min(A_i)}$ or $a_{\max(A_{i-1})}\fromEdge_c a_{\min(A_i)}$.%, where $\max(A_i)=a_{\max\{j\mid a_j\in A_i\}}$ and 
\end{itemize}
%{Note that for all $i$ all tuples $t_c^{i'}$ for which there are $k$ and $j$ such that $(c,i',k)=a_j$ and $j\in A_i$ represent the same element in $\bC$,in particular, they all have the same value for $c$. }
We write $A_{i}\toEdge_c A_{i+1}$ if $a_{\max(A_{i-1})}\toEdge_c a_{\min(A_i)}$ and  $A_{i}\fromEdge_c A_{i+1}$ if $a_{\max(A_{i-1})}\fromEdge_c a_{\min(A_i)}$. Note that 
\begin{align*}
|\{c\mid (c,i',k) \in \{a_j\mid j \in A_i\}\}|=1.\tag{$\ast$}
\end{align*} 
Define the primitive positive formula $\Phi(0,\dots,M')$ by adding conjuncts in the following way: 
For all $i$ and all $c$
\begin{itemize}
    \item if $A_{i}\toEdge_c A_{i+1}$, then add the conjunct $\phi_c(i,{i+1})$,
    \item if $A_{i}\fromEdge_c A_{i+1}$, then add the conjunct $\phi_c({i+1},{i})$,
    \item if $A_{i-1}\toEdge_c A_{i}\toEdge_{c+1} A_{i+1}$, then add all conjuncts of $\psi_c(i)$, and %for all $\phi(x)$ corresponding to a neighbour $(t,R)$ of $v_j$ in the incidence graph of $\bC$ that is not in $P$ add the conjunct $\phi({i})$, and
    \item if $A_{i-1}\fromEdge_{c+1} A_{i}\fromEdge_{c} A_{i+1}$, then add all  conjuncts of $\psi_c(i)$. %for all $\phi(x)$ corresponding to a neighbour $(t,R)$ of $v_j$ in the incidence graph of $\bC$ that is not in $P$ add the conjunct $\phi({i})$.
\end{itemize}
Note that in the last two cases, by definition of $\sim$,  \begin{align*}
[0,h+1]=\{i'\mid (c,i',k) \in \{a_j\mid j \in A_i\}\}.\tag{$\ast\ast$}
\end{align*} 
Denote the canonical database of $\Phi$ by $\bC'$. Note that $\bC'$ is a caterpillar with spine $(0,\dots,M')$.
Observe that, by $(\ast)$, {for all $i\in[0,M']$ and $j_1,j_2\in A_{i}$ with $a_{j_1}=(c,i',k)$ and $a_{j_2}=(d,j',\ell)$ we have that $c=d$. Hence $t_c^{i'}$ and $t_d^{j'}$ both represent the element $s_c$ in $\bC$. By definition of $A_i$ and $\sim$ we have that $(t_c^{i'})_k=(t_d^{j'})_{\ell}$. Hence, the map $\iota\colon [0,M']\to B$ that maps $i$ to $(t_c^{i'})_k$ for some (any) $(c,i',k)$ for which there is a $j\in A_{i}$ with $(c,i',k)=a_{j}$ is well defined.}   %i\mapsto t_{\max(A_i)}$.
The map $\iota$ is a satisfying assignment of $\Phi$ in $\bB$:
\begin{itemize}
    \item If $A_{i}\toEdge_c A_{i+1}$, then $(c,i',k)=a_{\max(A_i)}\toEdge_j a_{\min(A_{i+1})}=(d,j',\ell)$ and $\iota(i)=(t_c^{i'})_k$ and $\iota(i+1)=(t_d^{j'})_{\ell}$. By definition of $\toEdge_j$ we have $\bB\models\phi_c((t_c^{i'})_k,(t_d^{j'})_{\ell})$.
    \item The case $A_{i}\fromEdge_c A_{i+1}$ is analogous.
    \item If $A_{i-1}\toEdge_c A_{i}\toEdge_{c+1} A_{i+1}$ and $\psi(i)$ is a conjunct of $\psi_c(i)$, then, by (P2), there is a $i'\in[h]$ such that $\bB\models \phi((t_c^{i'})_k)$ for all $k\in[K]$. By $(\ast\ast)$ there is some $\ell\in[K]$ and some $j\in A_i$ such that $(c,i',\ell)=a_j$. Hence $\iota(i)=(t_c^{i'})_\ell$ and $\bB\models \psi(\iota(i))$.

    \item The case $A_{i-1}\fromEdge_{c+1} A_{i}\fromEdge_{c} A_{i+1}$ is analogous.
\end{itemize}
Therefore, $\iota$ can be extended to a homomorphism from $\bC'$ to $\bB$. 
Similar to the proof of Lemma~\ref{lem:MaltimpliesLAMDatalogIsSLAMDatalog} we can use Lemma~\ref{lem:MultiUnfoldingCaterpillar} to show that $\bC'$ is an unfolding of $\bC$ that fixes $v_1,\dots,v_n$. This concludes the proof.
%To derive a contradiction it remains to show that $\bC'$ is an unfolding of $\bC$ that fixes $v_1,\dots,v_n$. 
%By $(\ast)$ the map $\mu\colon[0,M']\to[0,m]$ that maps $i$ to the single element in $\{c\mid (c,i',k) \in \{a_j\mid j \in A_i\}\}$ is well defined.
%Note that that the map $i\mapsto s_{\mu(i)}$ can be extended to a homomorphism from $\bC'$ to $\bC$. 
%Let $m_1,\dots,m_{??}\in\N$ such that  
%Since alternating sum $\geq 1$ it is an easy exercise to show that there is an $i$ such that $m_i\leq m_{i-1},m_{i+1}$. 
%we only give a sketch of the proof here define sequence $(m_1,\dots,m_p)$ of numbers, von fixed edges getrennte invervalle anschauen, wenn in drei aufeinanderfolgenden zahlen die mittlere kleiner gleich den anderen ist, dann kann gefaltet werden, da summe im intervall position muss es so ein tripel geben
\end{proof}

It is possible extend the definition of indicator structures to minor conditions $\Sigma$ with multiple function symbols $f_1,\dots,f_n$ of arities $k_1,\dots,k_n$, respectively. Let $\bB$ be a structure. $\bB^{f_1,\dots,f_n}$ is the structure with domain $\{f_i(b_1,\dots,b_{k_i})\mid i\in[n], b_1,\dots,b_{k_i}\in B\}\}$ such that for each $i\in[n]$ the map $h_i\colon \bB^{k_i}\to\bB^{f_1,\dots,f_n}, t\mapsto f_i(t)$ is an embedding and $\bB^{f_1,\dots,f_n}$ is the disjoint union of the images of $h_1,\dots,h_n$. 
Let $\sim$ be the smallest equivalence relation on $B^{f_1,\dots,f_n}$ such that 
$f_i(a) \sim f_j(b)$ if $\Sigma$ contains $f_i(x_1,\dots,x_{k_i})\approx f_j(y_1,\dots,y_{k_j})$
and there is a map $s \colon \{x_1,\dots,x_{k_i},y_1,\dots,y_{k_j}\} \to B$ with 
$a = (s(x_1),\dots,s(x_{k_i}))$ and $b = (s(y_1),\dots,s(y_{k_j}))$. Then the \emph{indicator structure of $\Sigma$ with respect to $\bB$} is $\bB^{f_1,\dots,f_n}/_\sim$. 
Using these indicator structures the proof of Lemma~\ref{lem:unfoldedCPDualityImpliesSigmaImpliesHM} with (for finite structures) replaced by (for minions) works essentially in the same way, it just becomes even more technical. 
%Observe that the use of Lemma~\ref{lem:WlogSingleFunctionSymbol} to restrict to conditions $\Sigma$ with a single function symbol restriction to 

\section{Transitive tournaments have elevator chains}
To prove $(\ref{minor-p2})\Rightarrow(\ref{maltsev})$ we need the following statement. 

\begin{lemma}\label{lem:TnHasElevatorChains}
	Let $n\in\N$. Then $\Pol(\T_{n+1})$ contains operations that form an elevator chain of length $n$. 
\end{lemma}
This proof was discovered in discussion with Marlena Wünning.
\begin{proof}
	For $n=0$ the structure $\T_{0+1}$ consists of a single isolated point and has a constant polymorphism satisfying $\Ele0$.
	Now we consider the case $n\geq 1$. Let $\bI$ be the indicator structure of $\Ele n$ with respect to $\T_{n+1}$. Note that for every element $s$ of $\bI$ there are $a_1,\dots,a_n\in\T_{n+1}$ such that either
	\begin{itemize}
		\item $s$ is a singleton and there is an $i\in[n]$ and $a'_{i},a''_i\in\T_n$ with \[s=\{e_i(a_1,\dots,a_{i-1},a_{i},a'_{i},a''_i,a_{i+1},\dots,a_{n})\},\] $a_i\neq a'_i$, and $a'_i\neq a''_i$ or
		\item $s$ contains $e_0(a_1,\dots,a_n)$ and 
		\begin{align*}
			&e_i(a_1,\dots,a_{i-1},b,b,a_i,a_{i+1},\dots,a_n)\text{ and}\\ 
			&e_i(a_1,\dots,a_{i-1},a_i,b,b,a_{i+1},\dots,a_n)
		\end{align*} for all $b\in\T_{n+1}$ and all $i\in[n]$.
	\end{itemize}
	We define a map $f$ from $\bI$ to $[0,n]^n$. Elements of the first form are mapped to \[(a_1,\dots,a_{i-1},\min(a_i,a'_i,a''_i),a_{i+1},\dots,a_n)\] and elements of the second form are mapped to $(a_1,\dots,a_n)$. Clearly, $f$ is well defined.
	
	We define the digraph ${\vphantom{\tilde\T}}^1\tilde\T_{n+1}^n$ with domain $[0,n]^n$ that has an edge from $(a_1,\dots,a_n)$ to $(b_1,\dots,b_n)$ if and only if for all but at most one $j\in[n]$ we have $a_j<b_j$ and if $a_j\geq b_j$, then $a_j<n$ and $0<b_j$.
	
	Now we will show that $f$ is a homomorphism from $\bI$ to ${\vphantom{\tilde\T}}^1\tilde\T_{n+1}^n$. Let $(s,t)$ be an  edge in $\bI$ and %, then for all but at most one $j\in[n]$ we have $f(s)_j<f(t)_j$ and if $f(s)_j\geq f(t)_j$, then $n>f(s)_j$ and $f(t)_j>0$. 
	\begin{align*}
		e_i(a_1,\dots,a_{i-1},a_{i},a'_{i},a''_i,a_{i+1},\dots,a_{n})&\in s\text{ and}\\
		e_i(b_1,\dots,b_{i-1},b_{i},b'_{i},b''_i,b_{i+1},\dots,b_{n})&\in t
	\end{align*} be witnesses for this edge% $(s,t)$ in $\bI$. 
	Note that it is always possible to pick witnesses with $i\notin\{0,n+1\}$. Then $a_1<b_1,\dots,a_n<b_n,a'_i<b'_i$, and $a''_i<b''_i$. Hence $f(s)_j=a_j<b_j=f(t)_j$ for all $j\in[n]\setminus\{i\}$. Furthermore, $a_i<b_i$, $a'_i<b'_i$, and $a''_i<b''_i$ imply that $a_i,a'_i,a''_i<n$ and that $0<b_i,b'_i,b''_i$. Since $f(s)_i\in\{a_i,a'_i,a''_i\}$ and $f(t)_i\in\{b_i,b'_i,b''_i\}$ we get $f(s)_i<n$ and $0<f(t)_i$. Hence, $f$ is a homomorphism. 
	Note that it can be the case that $f(s)_j\geq f(t)_j$, if for example $a_i=0=a'_i$, $b'_i=n=b''_i$, and $n>a''_i\geq b_i>0$.
	
	Finally, we show ${\vphantom{\tilde\T}}^1\tilde\T_{n+1}^n\to\T_{n+1}$. %that there is a homomorphism from ${}^1\tilde\T_{n+1}^n$ to $\T_{n+1}$. 
	Since $(\{\P_{n+1}\},\T_{n+1})$ is a duality pair it suffices to show that $\P_{n+1}\not\to{\vphantom{\tilde\T}}^1\tilde\T_{n+1}^n$. Let $(a_1^0,\dots,a_n^0),\dots,(a_1^m,\dots,a_n^m)$ be a directed path in ${\vphantom{\tilde\T}}^1\tilde\T_{n+1}^n$. We want to show that $m\leq n$. By definition of the edge relation $a_1^1,\dots,a_n^1\geq 1$. Since in an edge all but at most one entry strictly increases we have that there are at least $n-j+1$ $i\in[n]$ such that $a_i^1<\dots<a_i^j$ for all $j\in[1,m]$. Hence for $j=n$ there is at least one $i\in[n]$ such that $a_i^n\geq n$. Since $a_i^n\in[0,n]$ we conclude $a_i^n=n$ and by definition of the edge relation $(a_1^n,\dots,a_n^n)$ does not have an outgoing edge. Therefore $m\leq n$ and $\P_{n+1}\not\to{\vphantom{\tilde\T}}^1\tilde\T_{n+1}^n$.
	In conclusion $\bI\to {\vphantom{\tilde\T}}^1\tilde\T_{n+1}^n\to\T_{n+1}$. Hence $\T_{n+1}\models\Ele n$.
\end{proof}
Note that ${\vphantom{\tilde\T}}^1\tilde\T_{n+1}^n$ is similar but not the same as the one-tolerant $n$-th power ${\vphantom{\T}}^1\mathord{\T}_{n+1}^n$ of $\T_{n+1}$ introduced by Larose, Loten, and Tardif in Section 4.2 in~\cite{LLT}. In fact, the smallest $k$ such that ${\vphantom{\T}}^1\mathord{\T}_{n+1}^k$ has a homomorphism to $\T_{n+1}$ is $n+1$.
We have seen in Lemma~\ref{lem:EleImpliesHM} that $\Ele n$ implies $\HM n$. In combination with Lemma~\ref{lem:TnHasElevatorChains} we obtain a new proof of the following well known fact. %and it is well known that $\bT_{n+1}\models\HM{n}$. 
\begin{corollary}
	Let $n\in\N$. Then $\Pol(\bT_{n+1})\models\HM n$.
\end{corollary}

\section{Proof of the main result}

Combining the previous section we can prove the main result.
\begin{proof}[Proof of Theorem~\ref{thm:main}]
	Let $n\in\N$ and let $\bB$ be a structure with a finite domain and a finite
	relational signature. 
	The implication $(\ref{maltsev})
	\Rightarrow(\ref{closedUnderFoldings})$ follows from Lemma~\ref{lem:elevatorChainsImplyClosedUndernFoldings} and Theorem~\ref{thm:mainOld}, the implication	
	$(\ref{closedUnderFoldings})
	\Rightarrow(\ref{caterpillarAndUnfolding})$ from Lemma~\ref{lem:closedUnderFoldingsImpliesClosedUnderUnfoldings}, the implication
	$(\ref{caterpillarAndUnfolding})
	\Rightarrow(\ref{caterpillar})$ is trivial, the implication	
	$(\ref{caterpillar})
	\Rightarrow(\ref{non-degenerate-minor})$ from Lemma~\ref{lem:unfoldedCPDualityImpliesSigmaImpliesHM}. 
	For the implication 	
	$(\ref{non-degenerate-minor})
	\Rightarrow(\ref{minor-p2})$ suppose that $\Sigma$ is a minor condition that holds in $\Pol(\bT_{n+2})$. By Theorem~\ref{thm:TnAreHMBlocker} we have $\T_{n+2}\not\models \HM n$. Hence, $\Sigma$ can not imply (for finite structures) $\HM n$. Therefore 
	$(\ref{non-degenerate-minor})$ implies $\bB\models\Sigma$, as desired.
	For the implication 
	$(\ref{minor-p2})
	\Rightarrow(\ref{maltsev})$ it suffices to show that $\Pol(\bT_{n+2})$ contains the desired operations. Since $(\{\bP_{n+2}\},\bT_{n+2})$ is a duality pair $\bT_{n+
	2}$ has caterpillar duality. Hence, by Theorem~\ref{thm:mainOld}, $\Pol(\bT_{n+2})$ contains $k$-absorptive
	operations of arity $\ell k$, for all $\ell, k \geq 1$. By Lemma~\ref{lem:TnHasElevatorChains} $\Pol(\bT_{n+2})$ contains elevator operations of length $n+2$.

	For the implication	
	$(\ref{caterpillarAndUnfolding})
	\Rightarrow(\ref{can-sym-lin-arc})$ let $\Pi_S$ be the canonical s$_n$lam Datalog program for $\bB$. Since $\bB$ has caterpillar duality Theorem~\ref{thm:mainOld} implies that the canonical lam Datalog program $\Pi$ for $\bB$ solves $\Csp(\bB)$.  Lemma~\ref{lem:MaltimpliesLAMDatalogIsSLAMDatalog} implies that $\Pi$ and $\Pi_S$ derive the goal predicate on the same instances of $\Csp(\bB)$. Therefore, $\Pi_S$ solves $\Csp(\bB)$.
	The implication	
	$(\ref{can-sym-lin-arc})
	\Rightarrow(\ref{sym-lin-arc})$ is trivial, the implication
	$(\ref{sym-lin-arc})
	\Rightarrow(\ref{caterpillarAndUnfolding})$ is Lemma~\ref{lem:SLAMimpliesUnfoldedCaterpillar}.

	The equivalence of $(\ref{minor-p2})-(\ref{pp-p2})$ follows from Theorem~\ref{thm:PPisMinorIsMinionHom}. That $(\ref{maltsev})\Leftrightarrow(\ref{lattice})$ follow from the equivalence of items (3) and (4) in Theorem~\ref{thm:mainOld} and the fact that the existence of an elevator chain of length $n$ is preserved under homomorphic equivalence. 
%	For the final comment assume that $\bB$ is a core. core kommentar beweisen
\end{proof}

Combining Lemma~\ref{lem:finiteCaterpillarDuality} and Theorem~\ref{thm:main} yields the following corollary.
\begin{corollary}
    If $\bB$ is a finite structure and $\mathcal F$ is a finite set of caterpillars such that $(\mathcal F,\B)$ is a duality pair, then items $(\ref{maltsev})-(\ref{lattice})$ hold for $\bB$, where 
    $n$ is the maximal number of vertebrae in a caterpillar in $\mathcal F$.
    %$\T_{n+2}$ pp-constructs $\bB$, where $n$ is the maximal number of vertebrae in a caterpillar in $\mathcal F$.
\end{corollary}

\section{Examples}\label{sec:examples}

Here we present some examples of interesting structures that can be pp-constructed from $\stCon$. 

\begin{example}\label{exa:aboveT4AndBelowT3}
	Let $\P_{313}$ be the digraph with vertices  $\{0,1,2,3,2',3',4',5'\}$ and edges $\{(0,1),(1,2),(2,3),(2',3),(2',3'),(3',4'),(4',5')\})$. Let $\mathcal F$ be the set of all unfoldings of $\P_{323}$ that fix the vertebrae indicated by the bold edges in Figure~\ref{fig:aboveT4AndBelowT3}. Observe that $(\mathcal F,\P_{313})$ is a duality pair. Hence $\P_{313}$ has 2-fixed unfolded caterpillar duality and, by Theorem~\ref{thm:main}, $\P_{313}$ is pp-constructable from $\bT_4$ but not from $\bT_3$. Furthermore, $\P_{313}\models\HM2$ and $\P_{313}\not\models\HM1$ (verified by computer). Hence, $\P_{313}$ can not pp-construct $\bT_4$ and it can pp-construct $\bT_3$. Therefore $\P_{313}$ lies strictly between $\bT_4$ and $\bT_3$.
	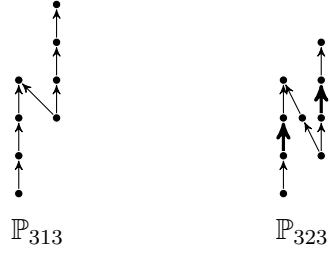
\begin{figure}
		\begin{tikzpicture}
			\node at (0,0) {
				\begin{tikzpicture}[scale = 0.5]
					\node[bullet] (1) at (0,1) {};
					\node[bullet] (2) at (0,2) {};
					\node[bullet] (3) at (0,3) {};
					\node[bullet] (4) at (0,4) {};
					
					\node[bullet] (3b) at (1,3) {};
					\node[bullet] (4b) at (1,4) {};
					\node[bullet] (5b) at (1,5) {};
					\node[bullet] (6b) at (1,6) {};
					\node at (0.5,0.0) {$\mathbb P_{313}$};
					\path[->,>=stealth'] 
					(1) edge (2)
					(2) edge (3)
					(3) edge (4)
					(3b) edge (4)
					(3b) edge (4b)
					(4b) edge (5b)
					(5b) edge (6b)
					;
			\end{tikzpicture}};
			\node at (3.5,0) {
				\begin{tikzpicture}[scale = 0.5]
					
					\node[bullet] (0) at (1,0) {};
					\node[bullet] (1) at (1,1) {};
					\node[bullet] (2) at (1,2) {};
					\node[bullet] (3) at (1,3) {};
					\node[bullet] (21) at (1.5,2) {};
					\node[bullet] (12) at (2,1) {};
					\node[bullet] (22) at (2,2) {};
					\node[bullet] (32) at (2,3) {};
					\node[bullet] (42) at (2,4) {};
					\node[bullet,opacity=0] (52) at (2,5) {};
					\node at (1.5,-1.0) {$\mathbb P_{323}$};
					
					\path[->,>=stealth'] 
					(0) edge (1)
					(1) edge[very thick] (2)
					(2) edge (3)
					(21) edge (3)
					(12) edge (21)
					(12) edge (22)
					(22) edge[very thick] (32)
					(32) edge (42)
					;
					
			\end{tikzpicture}};
		\end{tikzpicture}
		\centering
		\caption{An example of a graph (left) that is strictly between $\bT_4$ and $\bT_3$ and one of its obstructions (right).}\label{fig:aboveT4AndBelowT3}
	\end{figure}
\end{example}

\begin{example}\label{exa:aboveT5AndIncomparableWithT4}
    Let $\P_{414}$ be the digraph with vertices  $\{0,1,2,3,4,3',4',5',6',7'\}$ and edges $\{(0,1),(1,2),(2,3),(3,4),(3',4),(3',4'),(4',5'),(5',6'),(6',7')\})$. Let $\mathcal F$ be the set of all unfoldings of $\P_{423}$ and $\P_{324}$ that fix the vertebrae indicated by the bold edges in Figure~\ref{fig:aboveT5AndIncomparableWithT4}. Observe that $(\mathcal F,\P_{414})$ is a duality pair. Hence $\P_{414}$ has 3-fixed unfolded caterpillar duality and, by Theorem~\ref{thm:main}, $\P_{414}$ is pp-constructable from $\bT_5$ but not from $\bT_4$. Furthermore, $\P_{414}\models\HM2$ (verified by computer). Hence $\P_{414}$ can not pp-construct $\bT_4$. Therefore $\P_{414}$ can be pp-constructed from $\bT_5$ and is (pp-)incomparable with $\bT_4$.
    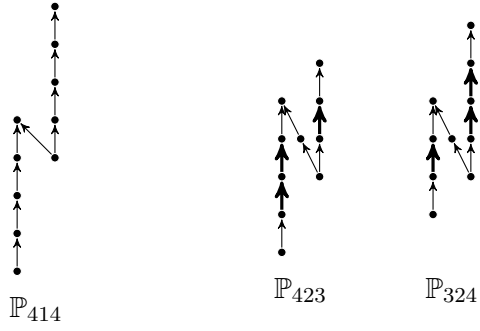
\begin{figure}
    	\begin{tikzpicture}
    		\node at (0,0) {
    		\begin{tikzpicture}[scale = 0.5]
   				\node[bullet] (0) at (0,0) {};
   				\node[bullet] (1) at (0,1) {};
   				\node[bullet] (2) at (0,2) {};
   				\node[bullet] (3) at (0,3) {};
   				\node[bullet] (4) at (0,4) {};
   				
   				\node[bullet] (3b) at (1,3) {};
   				\node[bullet] (4b) at (1,4) {};
   				\node[bullet] (5b) at (1,5) {};
   				\node[bullet] (6b) at (1,6) {};
   				\node[bullet] (7b) at (1,7) {};
   				\node at (0.5,-1.0) {$\mathbb P_{414}$};
   				\path[->,>=stealth'] 
   				(0) edge (1)
   				(1) edge (2)
   				(2) edge (3)
   				(3) edge (4)
   				(3b) edge (4)
   				(3b) edge (4b)
   				(4b) edge (5b)
   				(5b) edge (6b)
   				(6b) edge (7b)
   				;
    		\end{tikzpicture}};
    		\node at (3.5,0) {
    		\begin{tikzpicture}[scale = 0.5]
			
			\node[bullet] (m1) at (1,-1) {};
			\node[bullet] (0) at (1,0) {};
			\node[bullet] (1) at (1,1) {};
			\node[bullet] (2) at (1,2) {};
			\node[bullet] (3) at (1,3) {};
			\node[bullet] (21) at (1.5,2) {};
			\node[bullet] (12) at (2,1) {};
			\node[bullet] (22) at (2,2) {};
			\node[bullet] (32) at (2,3) {};
			\node[bullet] (42) at (2,4) {};
			\node[bullet,opacity=0] (52) at (2,5) {};
			\node at (1.5,-2.0) {$\mathbb P_{423}$};
			
			\path[->,>=stealth'] 
			(m1) edge (0)
			(0) edge[very thick] (1)
			(1) edge[very thick] (2)
			(2) edge (3)
			(21) edge (3)
			(12) edge (21)
			(12) edge (22)
			(22) edge[very thick] (32)
			(32) edge (42)
			;
			
			\end{tikzpicture}};
			\node at (5.5,0) {
			\begin{tikzpicture}[scale = 0.5]
					
			\node[bullet] (0) at (3,0) {};
			\node[bullet] (1) at (3,1) {};
			\node[bullet] (2) at (3,2) {};
			\node[bullet] (3) at (3,3) {};
			\node[bullet] (21) at (3.5,2) {};
			\node[bullet] (12) at (4,1) {};
			\node[bullet] (22) at (4,2) {};
			\node[bullet] (32) at (4,3) {};
			\node[bullet] (42) at (4,4) {};
			\node[bullet] (52) at (4,5) {};
			\node at (3.5,-2.0) {$\mathbb P_{324}$};
			\path[->,>=stealth'] 
			(0) edge (1)
			(1) edge[very thick] (2)
			(2) edge (3)
			(21) edge (3)
			(12) edge (21)
			(12) edge (22)
			(22) edge[very thick] (32)
			(32) edge[very thick] (42)
			(42) edge (52)
			;
   		\end{tikzpicture}};
    	\end{tikzpicture}
    	\centering
    	\caption{An example of a graph (left) that is strictly above $\bT_5$ and is incomparable with $\bT_4$ and some of its obstructions (right).}\label{fig:aboveT5AndIncomparableWithT4}
    \end{figure}
\end{example}

\begin{example}\label{exa:noElevatorTermsButHM}
    Let $\mathbb G$ be the digraph with vertices $\{0,1,1',2,2',3\}$ and edges \[\{(0,1),(1,2),(2,3),(0,3),(1',2),(1',3),(0,2'),(1,2'),(1',2')\})\] and $\P_4,\P_{323},\P_{32223},\dots$ be oriented paths as depicted in Figure~\ref{fig:noElevatorTermsButHM}. Observe that the tuple  $(\{\P_4,\P_{323},\P_{32223},\dots\},\mathbb G)$ is a duality pair. By Theorem~\ref{thm:mainOld} $\mathbb G$ is pp-constructable from st-Con. Furthermore, $\mathbb G\models\HM3$ (verified by computer). Hence $\mathbb G$ is strictly above st-Con.
    Note that $\Csp(\mathbb G)$ can be solved by a symmetric linear monadic Datalog Program but not by a s$_n$lam Datalog program. Hence, by Theorem~\ref{thm:main}, there is no $n$ such that $\bT_n$ pp-constructs $\mathbb G$. 
    \begin{figure}
    	\begin{tikzpicture}
    		\node at (0,0) {
    		\begin{tikzpicture}
    			\node[bullet,label=left:0] (0) at (0,0) {};
    			\node[bullet,label=left:1] (1) at (0,1) {};
    			\node[bullet,label=left:2] (2) at (0,2) {};
    			\node[bullet,label=left:3] (3) at (0,3) {};
    			
    			\node[bullet,label=right:1'] (1P) at (1,1) {};
    			\node[bullet,label=right:2'] (2P) at (1,2) {};
    			\node at (0.5,-0.5) {$\mathbb G$};
    			
    			\path[->,>=stealth'] 
    				(0) edge (1)
    				(1) edge (2)
    				(2) edge (3)
    				(0) edge[bend left=40] (3)
    				(1P) edge (2P)
    				(0) edge (2P)
    				(1) edge (2P)
    				(1P) edge (2)
    				(1P) edge (3)
    				;
    			\end{tikzpicture}};
    		\node at (6.5,0) {\begin{tikzpicture}[scale=0.5]
    				\node[bullet] (0) at (-2,0) {};
    				\node[bullet] (1) at (-2,1) {};
    				\node[bullet] (2) at (-2,2) {};
    				\node[bullet] (3) at (-2,3) {};
    				\node[bullet] (4) at (-2,4) {};
    				\node at (-2.0,-1) {$\mathbb P_4$};
    				\path[->,>=stealth'] 
    				(0) edge (1)
    				(1) edge (2)
    				(2) edge (3)
    				(3) edge (4)
    				;
    				
    				\node[bullet] (0) at (1,0) {};
    				\node[bullet] (1) at (1,1) {};
    				\node[bullet] (2) at (1,2) {};
    				\node[bullet] (3) at (1,3) {};
    				\node[bullet] (21) at (1.5,2) {};
    				\node[bullet] (12) at (2,1) {};
    				\node[bullet] (22) at (2,2) {};
    				\node[bullet] (32) at (2,3) {};
    				\node[bullet] (42) at (2,4) {};
    				\node at (1.5,-1.) {$\mathbb P_{323}$};
    				\path[->,>=stealth'] 
    				(0) edge (1)
    				(1) edge (2)
    				(2) edge (3)
    				(21) edge (3)
    				(12) edge (21)
    				(12) edge (22)
    				(22) edge (32)
    				(32) edge (42)
    				;
    				
    				\node[bullet] (0) at (5,0) {};
    				\node[bullet] (1) at (5,1) {};
    				\node[bullet] (2) at (5,2) {};
    				\node[bullet] (3) at (5,3) {};
    				\node[bullet] (21) at (5.5,2) {};
    				\node[bullet] (12) at (6,1) {};
    				\node[bullet] (22) at (6,2) {};
    				\node[bullet] (32) at (6,3) {};
    				\node[bullet] (23) at (6.5,2) {};
    				\node[bullet] (14) at (7,1) {};
    				\node[bullet] (24) at (7,2) {};
    				\node[bullet] (34) at (7,3) {};
    				\node[bullet] (44) at (7,4) {};
    				\node at (6,-1.0) {$\mathbb P_{32223}$};
    				\path[->,>=stealth'] 
    				(0) edge (1)
    				(1) edge (2)
    				(2) edge (3)
    				(21) edge (3)
    				(12) edge (21)
    				(12) edge (22)
    				(22) edge (32)
    				(23) edge (32)
    				(14) edge (23)
    				(14) edge (24)
    				(24) edge (34)
    				(34) edge (44)
    				;
    				
    				\node at (10.5,2) {$\dots$};
    		\end{tikzpicture}};
    	\end{tikzpicture}
    	\centering
    	\caption{An example of a graph (left) that is strictly above st-Con but can not be pp-constructed from any $\bT_n$ and its obstructions (right).}\label{fig:noElevatorTermsButHM}
    \end{figure}
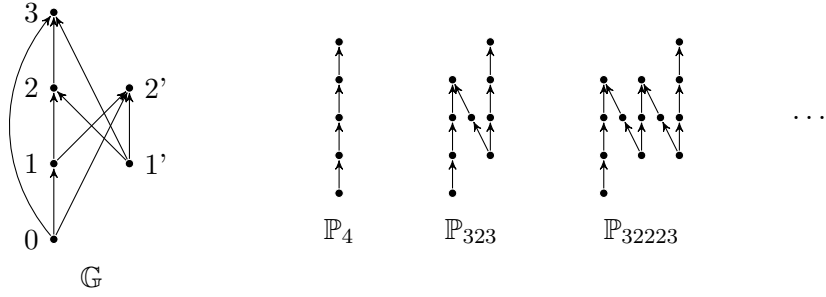
\end{example}

\section{Almost Symmetric Datalog}
In this section we will show that $n$-almost symmetric Datalog programs can solve the same problems as symmetric Datalog programs. Let $\Pi$ be an $n$-almost symmetric Datalog program. Note that the rule of that derives the goal predicate in a derivation of $\Pi$ is by definition not symmetric. The idea is to use high ary IDBs to apply all non-symmetric rules in parallel with the initial rule and then verify (using symmetric rules) that these applications were valid.
\begin{example}
    In Example~\ref{exa:TnSolvedBySlamDL} we have seen that $\Csp(\T_{n+2}')$ is solved by a s$_n$lam Datalog program. This CSP is also  solved by the symmetric arc Datalog program consisting of the rules
\begin{align*}
P(x_1',x_1,x_2',x_2,\dots,x_{n+1}',x_{n+1})\; &{:}{-} \;E(x,x_1')\wedge E(x_1,x_2')\wedge\dots\wedge E(x_n,x_{n+1}')\wedge E(x_{n+1},y) &\\
P(x_1',x_1,\dots,x_i'',x_i,\dots,x_{n+1}',x_{n+1})\; &{:}{-} \;P(x_1',x_1,\dots,x_i',x_i,\dots,x_{n+1}',x_{n+1})\wedge R(x_i',x_i'') &\\
P(x_1',x_1,\dots,x_i'',x_i,\dots,x_{n+1}',x_{n+1})\; &{:}{-} \;P(x_1',x_1,\dots,x_i',x_i,\dots,x_{n+1}',x_{n+1})\wedge R(x_i'',x_i') &
%\intertext{for $i\in[n+1]$ and}
     %G\; &{:}{-}\; P(x_1,x_1,\dots,x_{n+1},x_{n+1})&
\end{align*}
for $i\in[n+1]$ and the rule $G\; {:}{-}\; P(x_1,x_1,\dots,x_{n+1},x_{n+1})$.
\end{example}

Generalizing this construction we obtain the following lemma.
\begin{lemma}
    Let $n\in\N$ and let $\Pi$ be an $n$-almost symmetric Datalog program. Then there is a symmetric Datalog program $\Pi'$ such that $\Pi'$ can derive the goal predicate on a structure $\A$ if and only if $\Pi$ can derive the goal predicate on $\A$.
\end{lemma}
\begin{proof}
Let $\sim$ denote the equivalence relation induced by the quasi-order $\leq$ on the IDBs of $\Pi$. We construct the symmetric Datalog program $\Pi'$. For all IDBs $P_1\sim Q_1 <P_2\sim Q_2<\dots<P_k\sim Q_k$ of $\Pi$ we add the IDB $P_{P_1,Q_1,P_2,Q_2,\dots,P_{k},Q_k}$ to $\Pi'$, the arity of this new IDB is the sum of the arities of $P_1, Q_1 ,P_2, Q_2,\dots,P_k, Q_k$. Note that since $\Pi$ is $n$-almost symmetric we have that $k\leq n+1$. Hence, we added only finitely many IDBs to $\Pi'$. 
For each list $R_0,\dots,R_{k}$ of rules of $\Pi$ such that  
\begin{itemize}
    \item $R_0$ is of the form $P_{1}(\bar x_{1})\; {:}{-}\;\phi_0(\bar y_0)$,
    \item for $i\in[k-1]$ the rule $R_i$ is of the form $P_{i+1}(\bar x_{i+1}')\; {:}{-}\;Q_i(\bar x_i)\wedge \phi_i(\bar y_i)$, where $Q_i$ is the unique IDB in the body of $R_i$,
    \item $R_k$ is of the form $G\; {:}{-}\;Q_k(\bar x_k)\wedge \phi_k(\bar y_k)$, and
    \item $P_1\sim Q_1 <P_2\sim Q_2<\dots<P_k\sim Q_k$
\end{itemize}
we add the rule 
\begin{align*}
        P_{P_1,Q_1,\dots,P_k,Q_{k}}(\bar x_1',\bar x_1,\dots,\bar x_k',\bar x_k)\; &{:}{-}\; \phi_0(\bar y_0)\wedge \phi_1(\bar y_1)\wedge\dots\wedge\phi_k(\bar y_k)
\end{align*}
to $\Pi'$. We assume here that the sets of variables used in each rule are distinct.    
For each IDB $P_{P_1,Q_1,\dots,P_k,Q_{k}}$ of $\Pi'$ and each rule $P_{i}'(\bar x_{i}'')\; {:}{-}\;P_i(\bar x_i')\wedge \phi(\bar y_i')$ of $\Pi$, with $P_i'\sim P_i$, we add to $\Pi'$ the rule with the head $P_{P_1,Q_1,\dots,P_i',Q_i,\dots,P_k,Q_{k}}(\bar x_1',\bar x_1,\dots,\bar x_i'',\bar x_i,\dots,\bar x_k',\bar x_k)$ and the body
    \begin{align*}
        %P_{Q_1,\dots,Q_{k},P_1,\dots,P_i',\dots,P_{k}}(\bar x_1,\bar x_1',\dots,\bar x_i,\bar x_i'',\dots,\bar x_k,\bar x_k')\; &{:}{-}\; 
        P_{P_1,Q_1,\dots,P_i,Q_i,\dots,P_k,Q_{k}}(\bar x_1',\bar x_1,\dots,\bar x_i',\bar x_i,\dots,\bar x_k',\bar x_k)\wedge  \phi(\bar y_i')
    \end{align*}
where $\bar x_1',\bar x_1,\dots,\bar x_{i-1},\bar x_i,\bar x_{i+1}',\dots,\bar x_k',\bar x_k$ are tuples of new distinct variables. Note that since $P_i'\sim P_i$ we have that $P_{i}(\bar x_{i}')\; {:}{-}\;P_i'(\bar x_i'')\wedge \phi(\bar y_i')$ is also a rule $\Pi$. Hence, $\Pi'$ is, after adding all of these rules, symmetric.
Finally, we add for each IDB $P_{Q_1,Q_1,\dots,Q_k,Q_{k}}$ of $\Pi'$ the rule    
    \begin{align*}
        G\; &{:}{-}\; P_{Q_1,Q_1,\dots,Q_k,Q_{k}}(\bar x_1,\bar x_1,\dots,\bar x_k,\bar x_k)
    \end{align*}
    where $\bar x_1,\dots,\bar x_k$ are tuples of new distinct variables.
    Let $\A$ be a structure. It is easy to see that $\Pi'$ can derive the goal predicate on $\A$ with a derivation that derives $P_{Q_1,Q_1,\dots,Q_k,Q_{k}}$ on the tuple $(\bar a_1,\bar a_1,\dots,\bar a_k,\bar a_k)$ if and only if $\Pi$ can derive the goal predicate on $\A$ with a derivation that derives $Q_1$ on $\bar a_1$, $\dots$, and $Q_k$ on $\bar a_k$. Hence, $\Pi$ and $\Pi'$ can derive the goal predicate on the same structures.
\end{proof}

\section{Conclusions}
We introduced $n$-almost symmetric linear Datalog and show that this Datalog fragment can solve the same problems as the well known fragment of symmetric linear Datalog programs. When restricting additionally to monadic IDBs, then the Datalog fragments differ in expressibility. 
We characterized the structures whose CSP is solvable by s$_n$lam Datalog as those that are pp-constructable from $\T_n$. Recall that pp-constructability orders the transitive tournaments in the following way
\[\stCon<\dots<\T_4<\T_3<\T_2<\T_1.\]
In Section~\ref{sec:examples} we have seen that not all structures pp-constructable from $\stCon$ are pp-interconstructable with some $\T_n$. Based on these examples we pose the following conjecture. 
\begin{conjecture}
    A structure $\bB$ lies strictly above $\stCon$ if and only if there are $n,m\in\N$ such that $\Csp(\bB)$ is solved by an $n$-almost symmetric linear monadic Datalog program, where for every rule the canonical database of its body is a directed path of length at most $m$. % (a$_m$ should only allow paths otherwise you can solve $\bB_n$)
\end{conjecture}
Another avenue of future research is starting again with slam Datalog and replacing linearity by \emph{$n$-almost linearity}.  We expect this to yield Datalog fragments for CSPs of the structures pp-constructed from \[\bB_n=(\{0,1\},\{1\},\{0,1\}^n\setminus\{(1,\dots,1)\}).\] It is well known that these structures form a chain 
\[\bB_\infty<\dots<\bB_{3}<\bB_{2}\]
and that $\bB_2$ is pp-interconstructable with $\T_3$. This would be a next step in characterizing all structures that are pp-constructable from a structure with finite duality. %mention duality  pair or end here?

\begin{acknowledgements*} 
The authors would like to thank Andrew Moorhead for many helpful and motivating discussions.
\end{acknowledgements*}

\bibliography{citations} 
\bibliographystyle{alpha}
\end{document}